\documentclass{lmcs} 


\usepackage{amsmath,amssymb}
\usepackage{hyperref}

\usepackage{enumitem}
\usepackage{xspace}
\usepackage{footmisc}
\usepackage{tikz}
\usetikzlibrary{arrows.meta,shapes}

\newcommand*\defd[2][\em\color{red!25!black}]{{#1#2}}
\newcommand*\strong[2][\bfseries\mathversion{bold}]{{#1#2}}

\newcommand{\arity}{\operatorname{ar}}
\newcommand*{\fullsucc}[3][\to]{\ensuremath{#2\text{$#1$}#3}}
\def\nfullsucc{\fullsucc[\not\to]}
\newcommand{\FO}{\ensuremath{\operatorname{FO}}\xspace}
\newcommand{\MSO}{\ensuremath{\operatorname{MSO}}\xspace}
\newcommand{\CMSO}[1][]{\ensuremath{\C[#1]{\!\!}\MSO}\xspace}

\newcommand{\C}[1][]{\ensuremath{\operatorname{C}_{#1}}\xspace}

\newcommand{\leaf}{\mathsf{leaf}}
\newcommand{\rootT}{\mathsf{root}}

\newcommand{\ancestor}{\mathsf{ancestor}}
\newcommand{\children}{\mathsf{children}}
\newcommand{\betweeness}{\mathsf{betweeness}}
\newcommand{\cross}{\mathsf{cross}}

\newcommand{\node}{\ensuremath{\mathsf{node}}\xspace}
\DeclareMathOperator*\symdiff{\bigtriangleup}
\newcommand\struct[1]{\ensuremath{\langle#1\rangle}}

\newcommand{\copyV}[1][]{\ensuremath{\mathsf{copy}_{#1}}\xspace}
\newcommand{\colorV}[1][]{\ensuremath{\mathsf{color}_{#1}}\xspace}
\newcommand{\parent}{\ensuremath{\mathsf{parent}}\xspace}
\newcommand{\edge}{\ensuremath{\mathsf{edge}}\xspace}
\newcommand{\medge}{\ensuremath{\textsf{m}\text-\mathsf{edge}}\xspace}
\newcommand{\tedge}{\ensuremath{\mathsf{t}\text-\mathsf{edge}}\xspace}
\newcommand{\cedge}{\ensuremath{\mathsf{c}\text-\mathsf{edge}}\xspace}
\newcommand{\redge}{\ensuremath{\mathsf{r}\text-\mathsf{edge}}\xspace}
\newcommand{\epsord}{\ensuremath{\mathsf{ord}}\xspace}
\newcommand{\SET}{\ensuremath{\mathsf{SET}}\xspace}
\newcommand{\SSET}{\ensuremath{\mathsf{STRS}}\xspace}
\newcommand{\str}{\ensuremath{\mathsf{str}}\xspace}
\newcommand{\BIPARTITION}{\ensuremath{\mathsf{BIPART}}\xspace}
\newcommand{\splitSet}{\ensuremath{\mathcal{B}_\mathsf{split}^G}}
\newcommand{\joinSet}{\ensuremath{\mathcal{B}_\mathsf{join}^G}}
\newcommand{\marker}{\ensuremath{\mathsf{marker}}\xspace}

\newcommand{\descendant}{\ensuremath{\mathsf{desc}}\xspace}
\newcommand{\child}{\ensuremath{\mathsf{child}}\xspace}
\NewDocumentCommand{\repr}{o d()}{\ensuremath{{\mathsf{repr}\IfValueT{#1}{_{#1}}\IfValueT{#2}{(#2)}}}\xspace}

\newcommand{\leafset}{\ensuremath{\mathsf{leafset}}\xspace}

\newcommand\DEGENERATE{\ensuremath{\mathsf{degenerate}}\xspace}
\newcommand\PRIME{\ensuremath{\mathsf{prime}}\xspace}
\newcommand\LINEAR{\ensuremath{\mathsf{linear}}\xspace}
\newcommand\SERIES{\ensuremath{\mathsf{series}}\xspace}
\newcommand\PARALLEL{\ensuremath{\mathsf{parallel}}\xspace}

\newcommand{\inNeighbours}{\operatorname{in}}
\newcommand{\outNeighbours}{\operatorname{out}}

\newcommand{\etc}{\emph{etc\ldots}\xspace}
\newcommand{\resp}{resp.\xspace}
\newcommand{\aka}{aka\xspace}
\newcommand{\card}[1]{\ensuremath{\left|#1\right|}}
\newcommand{\ifof}{if and only if\xspace}
\renewcommand{\implies}{\rightarrow}

\def\eg{{\em e.g.}}
\def\ie{{\em i.e.}}

\begin{document}

\title[CMSO-transductions for tree-like graph decompositions]{CMSO-transductions for tree-like graph decompositions\rsuper*}
\titlecomment{{\lsuper*}An extended abstract of this work was presented at the 42nd International Symposium on Theoretical Aspects of Computer Science (STACS 2025)~\cite{STACSversion}.}
\thanks{
}	

\author[R.~Campbell]{Rutger Campbell}[a]
\author[B.~Guillon]{Bruno Guillon\lmcsorcid{https://orcid.org/0000-0003-1630-3404}}[b]
\author[M.~M.~Kant\'{e}]{Mamadou Moustapha Kant\'{e}\lmcsorcid{https://orcid.org/0000-0003-1838-7744}}[b]
\author[E.~J.~Kim]{Eun Jung Kim\lmcsorcid{https://orcid.org/0000-0002-6824-0516}}[c,d,e]
\author[N.~K\"{o}hler]{Noleen K\"{o}hler\lmcsorcid{https://orcid.org/0000-0002-1023-6530}}[f]

\address{University of Waterloo, Waterloo, Canada}	
\email{rtrcampb@uwaterloo.ca}  

\address{Université Clermont Auvergne, Clermont Auvergne INP, LIMOS, CNRS, Clermont-Ferrand, France}	
\email{bruno.guillon@uca.fr, mamadou.kante@uca.fr} 

\address{KAIST, Daejeon, South Korea}	
\email{eunjung.kim@kaist.ac.kr}  

\address{CNRS, France}	 
\address{Discrete Mathematics Group, Institute for Basic Science, Daejeon, Korea}	
\email{} 

\address{University of Leeds, Leeds, UK}	
\email{N.Koehler@leeds.ac.uk}  





\begin{abstract}
  \noindent We give \CMSO-transductions that, given a graph $G$, output its modular decomposition, its split decomposition and its bi-join decomposition. This
  improves results by Courcelle [Logical Methods in Computer Science, 2006] who gave such transductions using order-invariant \MSO, a strictly more expressive
  logic than \CMSO.  Our methods more generally yield \CMSO[2]-transductions that output the canonical decompositions of weakly-partitive set systems and
  weakly-bipartitive systems of bipartitions.
\end{abstract}

\maketitle

\section{Introduction}
A decomposition of a graph, especially a tree-like decomposition, is a result of recursive separations of a graph
and is extremely useful for investigating combinatorial properties such as colourability, and for algorithm design. Such a decomposition also plays a fundamental role when one wants to understand the relationship between logic and a graph class.
Different notions of the complexity of a separation motivate different ways to decompose, such as tree-decomposition, branch-decomposition, rank-decomposition and carving-decomposition.
Furthermore, some important graph classes can be defined through the tree-like decomposition they admit;
cographs with cotrees and distance-hereditary graphs with split decompositions being prominent examples.

For a logic $\mathsf{L}$,
an $\mathsf{L}$-transduction 
is a non-deterministic map from a class of relational structures to a new class of relational structures using $\mathsf{L}$-formulas. Transduction that given a graph $G$ output 
a tree-like decomposition of $G$ is of particular interest.
Notably, obtaining a decomposition of a graph via a transduction implies that any property that is definable using a decomposition, is also definable on graphs that admit such a
decomposition. To provide an example, a transduction constructing a modular decomposition of a graph $G$ can be used to define a sentence determining that
  the number of modules in $G$ is even (see \autoref{sec:modular-dec} for definitions and \autoref{ex:numModules} for details).  
Moreover, tree-like decompositions can be often represented by labelled trees, for which the equivalence of recognisability by a tree automaton and definability in \MSO with modulo counting predicates, denoted \CMSO, is well known~\cite{TMSOLOG1}.
Hence, it is an interesting question to consider what kind of graph decompositions can be transduced from an input graph using $\mathsf{L}$-transductions for some extension
$\mathsf{L}$ of \MSO. 

In a series of papers~\cite{TMSOLOG1,TMSOLOF5,Courcelle96,Courcelle06}, Courcelle investigated the relationship between the graph properties that can be defined in an extension of \MSO and the graph properties that can be recognized by a tree automaton where the tree-automaton receives a term representing the input graph. 
In particular, Courcelle's theorem states that any graph property that is definable in the logic \CMSO can be recognized by a tree automaton  receiving a
tree-decomposition of small width \cite{TMSOLOG1}.
Combining this result with the linear time algorithm for computing tree-decompositions \cite{bodlaender1993linear}, yields that \CMSO\ model-checking can be done in linear time on graphs of bounded treewidth.
The converse statement -- whether recognisability by a tree automaton implies definability in \CMSO\ on graphs of bounded treewidth --
was conjectured by Courcelle in \cite{TMSOLOG1} and finally settled by Boja\'nczyk and Pilipczuk \cite{BojanczykP16}.
  The key step to obtain this result
is obtaining a tree-decomposition of a graph via an \MSO-transduction, a strategy which was initially proposed in \cite{TMSOLOF5} and is now
standard. 

The obvious next question is whether an analogous result can be proved for graphs of bounded clique-width and for more general combinatorial objects, most notably, matroids representable over a fixed finite field and of bounded branch-width. Due to the above-mentioned strategy, the key challenge is to produce corresponding tree-like decompositions by \MSO-transduction. It is known that clique-width decompositions can be \MSO-transduced for
graphs of bounded linear clique-width \cite{BojanczykGP21}. However, it is unknown if clique-width decompositions can be \MSO-transduced from a graph in general. In fact,
this question remained open even for distance-hereditary graphs, which are precisely graphs of rank-width 1 (thus, of constant clique-width). 

Besides tree-decompositions, the problem of obtaining the cotree of a cograph via transduction, and in general
hierarchical decompositions such as modular
decompositions and split decompositions were considered in the literature~\cite{Courcelle96,Courcelle99,Courcelle06,Courcelle13}. In \cite{Courcelle96}, Courcelle provides transductions, using order-invariant \MSO, for cographs and modular decompositions of
graphs of bounded modular width. Order-invariant \MSO allows the use of a linear order of the set of vertices and is more expressive than \CMSO
\cite{GanzowR08}. The applicability of these transductions was later generalized using the framework of weakly-partitive set systems\footnote{Weakly-partitive set systems are set systems enjoying some nice closure properties, which were then used
to show that some set systems allow canonical tree representations, see for instance the thesis by Montgolfier and Rao
\cite{montgolfierThesis,raoThesis}.} to obtain order-invariant
\MSO-transductions of modular and split decompositions \cite{Courcelle06}. It was left as an open question whether one can get rid of the order (see for instance
\cite{CourcelleT12} where an overview of the result on hierarchical decompositions was given).

\subsection{Our results}
In this paper,
we consider decompositions given by nested partitions.
We view partitions of a given kind as a `set system'.
A set system consists of a set $U$, the universe, and a set $\mathcal{S}$ of subsets of $U$. Two sets overlap if they have non-empty intersection but neither of them is contained in the other. If no two elements in a set system $(U,\mathcal{S})$ overlap, i.e. the set system is \emph{laminar}, then the subset relation in $(U,\mathcal{S})$ can be described by a rooted tree $T$, called the \emph{laminar tree} of $(U,\mathcal{S})$.
For any set system $(U, \mathcal{S})$ we can look at the subset of strong sets, i.e., sets that do not overlap with any other set, and the laminar tree $T$ they induce. 

Our main contribution is providing a transduction that given a laminar set system outputs its laminar tree (see \autoref{thm:transduce laminar tree}). To obtain this theorem, the main technical challenge is  equipping each inner node of the laminar tree with a leaf representative such that the representative relation is definable in MSO using an even cardinality set predicate while guaranteeing that each leaf serves as representative to a bounded number of inner nodes.

Given a graph $G$ we can consider the set system $(V(G),\mathcal{M})$ where $\mathcal{M}$ is the set of all modules in $G$. We obtain the modular decomposition, or the cotree in case $G$ is a cograph, by equipping the laminar tree of the suitable set system with some additional structure. The additional structure allows us to recover the graph from the respective decomposition.

Abstractly, the set systems mentioned are instances of \emph{weakly-partitive set systems}.
Roughly speaking, if a set system $(U,\mathcal{S})$ is well behaved, i.e. $(U,\mathcal{S})$ is weakly-partitive (definition in \autoref{sec:preliminaries}), then there is a tree $T$, a labelling $\lambda$ of $T$ and a partial order $<$ of its nodes such that $(U,\mathcal{S})$ is completely described by $(T,\lambda,<)$ \cite{chein1981partitive,raoThesis}. See \autoref{fig:partitive} for an example.
\begin{figure}
    \includegraphics[scale=0.8]{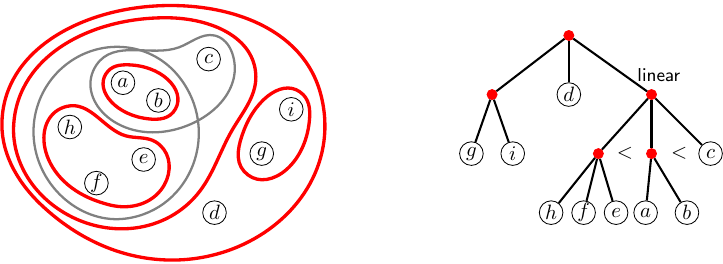}%
    \caption{%
        Left: A weakly-partitive set system for which strong sets are indicated by red, thick lines and singletons are omitted. Right: The laminar tree obtain by considering the laminar set system consisting of the strong sets. Here the node labelled \LINEAR plus the linear ordering of its children indicate that the leaves of every $<$-interval correspond to a set in the set system (\eg~$\{a,b,c\}$ is in the set system while $\{h,f,e,c\}$ is not). Note that some labels of nodes are omitted.%
    }%
    \label{fig:partitive}
\end{figure}
\bigbreak

We show the following, wherein each item $\lambda$ is a suitable labelling of the nodes of the laminar tree $T$, $<$ is a partial ordering of its nodes, and $F$
is an additional edge relation defined only on pairs of siblings in $T$. A visualization how results depend on each other is given in
\autoref{fig:overview}.

\begin{thm}\label{thm:summary}
    There are non-deterministic \CMSO[2]-transductions $\tau_1,\ldots,\tau_4$ such that: 
    \begin{enumerate}
        \item 
           For any laminar set system $(U,\mathcal{S})$, $\tau_1(U, \mathcal{S})$ is non-empty
           and every output in~$\tau_1(U, \mathcal{S})$ is a laminar tree $T$ of $(U,\mathcal{S})$ (\autoref{thm:transduce laminar tree});
        \item 
            For any weakly-partitive set system $(U,\mathcal{S})$, $\tau_2(U,\mathcal{S})$ is non-empty
            and every output in~$\tau_2(U,\mathcal{S})$ is a weakly-partitive tree $(T,\lambda,<)$ of $(U,\mathcal{S})$ (\autoref{thm:transduce weakly-partitive tree});
        \item 
            For any graph $G$, $\tau_3(G)$ is non-empty
            and every output in~$\tau_3(G)$ is a modular decomposition $(T,F)$ of $G$ (\autoref{thm:transduce modular decomposition});
        \item 
            For any cograph $G$, $\tau_4(G)$ is non-empty
            and every output in~$\tau_4(G)$ is a cotree $(T,\lambda)$ of $G$ (\autoref{thm:transduce cotree}).
    \end{enumerate}
\end{thm}

Other tree-like graph decompositions can be obtained by considering systems of bipartitions. Aiming to study such decompositions, we can apply our techniques to systems of bipartitions. A systems of bipartitions consists of a universe $U$ and a set $\mathcal{B}$ of bipartitions of $U$. Two bipartitions of $U$ overlap if neither side of one of the bipartition is contained in either side of the other bipartition. In case $(U,\mathcal{B})$ has no overlapping bipartitions, then $(U,\mathcal{B})$ can be described by an undirected tree, also called \emph{laminar tree}, in which bipartitions correspond to edge cuts. We can define the concept of strong bipartitions equivalently and consider the laminar tree induced by the strong bipartitions in $(U,\mathcal{B})$.

Given a graph $G$, we consider the system of bipartitions $(V(G), \mathcal{S})$ where $\mathcal{S}$ contains all splits in $G$, or the
system of bipartitions $(V(G), \mathcal{B})$ where $\mathcal{B}$ contains all bi-joins in $G$. Equipping the laminar tree of the respective systems of bipartitions with additional structure yields the split decomposition or bi-join decomposition of the graph. These systems of bipartitions are examples of \emph{weakly-bipartitive systems of bipartitions}. Similar to weakly-partitive set systems, we can completely describe any weakly-bipartitive system of biparitions by a tree $T$, a labelling $\lambda$ and linear order $<_t$ of the children of some particular nodes $t$  \cite{montgolfierThesis}. See \autoref{fig:bipartitive} for an example.
\begin{figure}
    \includegraphics[scale=0.8]{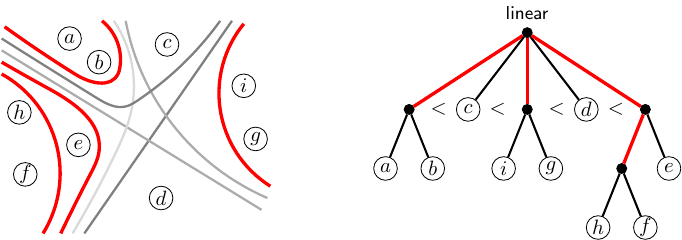}%
    \caption{%
        Left: A weakly-bipartitive system of bipartitions for which strong bipartitions are indicated by red, thick lines (the strong bipartitions are $\{\{a,b\},\{c,d,e,f,g,h,i\}\}$, $\{\{h,f\},\{a,b,c,d,e,g,i\}\}$, $\{\{h,f,e\},\{a,b,c,d,g,i\}\}$, $\{\{i,g\},\{a,b,c,d,e,f,h\}\}$ plus the bipartition $\{\{x\}, U\setminus \{x\}\}$ for every $x\in \{a,b,c,d,e,f,g,h,i\}$) and bipartitions of the form $\{\{x\},U\setminus \{x\}\}$ are omitted. Right: The laminar tree obtain by considering the laminar system of bipartitions consisting of the strong bipartitions. The node labelled \LINEAR plus the linear ordering of its children indicate that the leaves of every $<$-interval correspond to a bipartition in the system of biparitions (\eg~$\{\{a,b,c,i,g\},\{d,e,f,h\}\}$ is in the system of bipartitions while $\{\{a,b,d\},\{c,e,f,g,h,i\}\}$ is not). Note that some labels of nodes are omitted.%
    }%
    \label{fig:bipartitive}
\end{figure}
\bigbreak

We prove the following theorem, where similarly to before each item $\lambda$ is a suitable labelling of the nodes of the laminar tree $T$, $<$ is a
  partial ordering of its nodes, and $F$ is an additional edge relation defined only on pairs of siblings in $T$.

\begin{thm}\label{thm:summary}
    There are non-deterministic \CMSO[2]-transductions $\tau_1,\ldots,\tau_3$ such that: 
    \begin{enumerate}
         \item 
            For any weakly-bipartitive set system $(U,\mathcal{B})$, $\tau_1(U,\mathcal{B})$ is non-empty
            and every output in~$\tau_1(U,\mathcal{B})$ is a weakly-bipartitive tree $(T,\lambda,<)$ of $(U,\mathcal{B})$ (\autoref{thm:transduce weakly-bipartitive tree});
        \item 
             For any graph $G$, $\tau_2(G)$ is non-empty
             and every output in~$\tau_2(G)$ is a split decomposition $(T,F)$ of $G$ (\autoref{thm:split});
        \item 
             For any graph $G$, $\tau_3(G)$ is non-empty
             and every output in~$\tau_3(G)$ is a bi-join decomposition $(T,F)$ of $G$ (\autoref{thm:bi-join1}). 
    \end{enumerate}
\end{thm}
The key step in obtaining these transductions
is to transduce the laminar tree $T$ given a set system $(U,\mathcal{S})$.
The crux here is to find a suitable representative of each node of $T$ among the elements of $U$
and a non-deterministic colouring which allows the assignment of representatives to nodes by means of a \CMSO[2]-formula.  
It should be mentioned that a similar result is claimed in the preprint \cite{Boj23},
where a proof sketch designing a \CMSO[3]-transduction is described. 
Once the laminar tree is obtained,
the additional relations for each decomposition
can be obtained using a deterministic \MSO-transduction.
Notice that for each of these transductions,
there exists an inverse deterministic \MSO-transduction,
namely a transduction that from the tree-like decomposition outputs the original structure. It is worth mentioning that a corollary of
\autoref{thm:summary}(4) is a \CMSO-transduction outputting a rank-decomposition of width $1$ for every graph of rank-width $1$.

\begin{figure}
    \centering
    \begin{tikzpicture}[gnode/.style={draw, fill=white, inner sep=3pt, minimum width=0.4cm, execute at begin node=\setlength{\baselineskip}{0.8em}}]

        \tikzstyle{ns1}=[line width=1]

        \fill [gray!10] (-1.8,-0.7) rectangle (1.8,1.4);
        \node at (0,1){\textcolor{orange}{\autoref{sec:transduce laminar tree}}};
        \fill [gray!10] (-6.9,-6.4) rectangle (-0.1,-1.1);
        \fill [gray!10] (0.3,-6.4) rectangle (7.3,-1.1);
        \node at (-5.7,-2){\textcolor{orange}{\autoref{sec:transduce modular decomposition}}};
        \node at (5.7,-2){\textcolor{orange}{\autoref{sec:split}}};

        \node[gnode,text width = 2.6cm] (a) at (0,0){\autoref{thm:transduce laminar tree}\phantom{blablabla}
        \scriptsize \textcolor{teal}{laminar set system $\rightarrow$ laminar tree}};

        \node[gnode,text width = 2.1cm] (b) at (-3.2,-2){\autoref{lem:transduce laminar tree induced by weakly-partitive family}\phantom{blablabla}
        \scriptsize \textcolor{teal}{weakly-partitive set system $\rightarrow$ laminar tree}};

        \node[gnode,text width = 2.35cm] (c) at (3,-2){\autoref{lem:transduce laminar tree induced by weakly-bipartitive family}\phantom{blablabla}
        \scriptsize \textcolor{teal}{weakly-bipartitive set system $\rightarrow$ laminar tree}};

        \node[gnode,text width = 2.72cm] (d) at (-5.25,-4.2){\autoref{thm:transduce weakly-partitive tree}\phantom{blablabla}
        \scriptsize \textcolor{teal}{weakly-partitive set system $\rightarrow$ weakly-partitive tree}};

        \node[gnode,text width = 2.4cm] (e) at (-1.6,-4.2){\autoref{thm:transduce modular decomposition}\phantom{blablabla}
        \scriptsize \textcolor{teal}{graph $\rightarrow$ modular decomposition}};

        \node[gnode,text width = 2.97cm] (g) at (2.1,-4.2){\autoref{thm:transduce weakly-bipartitive tree}\phantom{blablabla}
        \scriptsize \textcolor{teal}{weakly-bipartitive set system $\rightarrow$ weakly-bipartitive tree}};

        \node[gnode,text width = 2.05cm] (h) at (6,-4.2){\autoref{thm:split}\phantom{blablabla}
        \scriptsize \textcolor{teal}{graph $\rightarrow$ split decomposition}};

        \node[gnode,text width = 2.4cm] (f) at (-3.2,-5.8){\autoref{thm:transduce cotree}\phantom{blablabla}
        \scriptsize \textcolor{teal}{cograph $\rightarrow$ cotree}};

        \node[gnode,text width = 2.3cm] (i) at (5.2,-5.65){\autoref{thm:bi-join1} 
        \scriptsize \textcolor{teal}{graph $\rightarrow$ bi-join  decomposition}};

        \draw[ns1, -Stealth] (a) -- (b);
        \draw[ns1, -Stealth] (b) -- (c);
        \draw[ns1, -Stealth] (b) -- (d);
        \draw[ns1, -Stealth] (b) -- (e);
        \draw[ns1, -Stealth] (b) -- (f);
        \draw[ns1, -Stealth] (c) -- (g);
        \draw[ns1, -Stealth] (c) -- (h);
        \draw[ns1, -Stealth] (c) -- (i);

    \end{tikzpicture}
    \caption{%
        Overview of the various transductions in the paper. An arrow from $x$ to $y$ indicates that result $x$ is used in the proof of result $y$.
    }%
    \label{fig:overview}
\end{figure}
\bigbreak

\subsection{Organization}
The paper is organized as follows. In \autoref{sec:preliminaries} we introduce terminology and notation needed. In \autoref{sec:transduce laminar tree} we prove \autoref{thm:transduce laminar tree}. In \autoref{sec:transduce modular decomposition} we provide the proof of \autoref{thm:transduce weakly-partitive tree} and \autoref{thm:transduce modular decomposition} and obtain \autoref{thm:transduce cotree}. In \autoref{sec:split} we provide the proofs of \autoref{thm:transduce weakly-bipartitive tree}, \autoref{thm:split} and \autoref{thm:bi-join1}.

\section{Preliminaries}\label{sec:preliminaries}
\subsection{Graphs, trees, set systems}
\subparagraph{Graphs}
We use standard terminology of graph theory,
and we fix some notations.
We consider graphs to be finite.
Given a \emph{directed graph}~$G$,
its sets of \emph{vertices} and \emph{edges}
are denoted by \defd{$V(G)$} and \defd{$E(G)$}, respectively.
We denote by \defd{uv} an edge~$(u,v)\in E(G)$.
An undirected graph is no more than a directed graph
for which~$E(G)$ is symmetric (\ie, $uv\in E(G)\iff vu\in E(G)$).
The notions of \emph{paths}, \emph{connected components}, \etc are defined as usual.
Given a subset~$Z$ of~$V(G)$,
we denote by \defd{$G[Z]$} the sub-graph of~$G$ induced by~$Z$. 

\subparagraph{Trees}
A \defd{tree} is a connected undirected graph without cycles.
In the context of trees,
we use a slightly different terminology than for graphs.
In particular, vertices are called~\defd{nodes},
nodes of degree at most~$1$ are called \defd{leaves},
and nodes of degree greater than~$1$ are called \defd{inner}.
The set of leaves is denoted~\defd{$L(T)$};
thus the set of inner nodes is~$V(T)\setminus L(T)$. For a node $t$ of a tree $T$ and a neighbour $s$ of $t$, we denote by $T^t_s$ the connected component of $T-t$ containing $s$.
We sometimes consider \defd{rooted trees},
namely trees with a distinguished node, called the \defd{root}.
Rooted trees enjoy a natural orientation of their edges
toward the root,
which induces the usual notions of
\defd{parent}, \defd{child}, \defd{sibling}, \defd{ancestor} and \defd{descendant}.
Hence, we represent a rooted tree by a set of nodes with an ancestor/descendant relationship (instead of specifying the root).
We use the convention that every node is one of its own ancestors and descendants.
We refer to ancestors (\resp descendants) of a node
that are not the node itself
as \defd{proper ancestors} (\resp \defd{proper descendants}).
For a node~$t$ of a rooted tree~$T$,
we denote by~\defd{$T_t$} the subtree of~$T$ rooted in~$t$
(\ie, the restriction of~$T$ to the set of descendants of~$t$).

\subparagraph{Set systems and laminar trees}
A \defd{set system} is a pair~$(U,\mathcal{S})$
where~$U$ is a finite set, called the \defd{universe},
and~$\mathcal{S}$ is a family of subsets of~$U$
where~$\emptyset\notin\mathcal{S}$,
$U\in\mathcal{S}$,
and~$\{a\}\in\mathcal{S}$ for every~$a\in U$.%
\footnote{%
    Though these restrictions on~$\mathcal{S}$ are not usual for set systems,
    it is convenient for our contribution
    and it does not significantly impact the generality of set systems:
    every family~$\mathcal{F}$ of subsets of~$U$
    can be associated with a set system~$(U,\mathcal{S})$
    where $\mathcal{S}=\big(\mathcal{F}\setminus\{\emptyset\}\big)\cup\big\{U\big\}\cup\big\{ \{a\} \mid a\in U\big\}$.%
}
Two sets~$X$ and~$Y$ \defd{overlap}
if they are neither disjoint nor related by containment.
A set system~$(U,\mathcal{S})$
is said to be \defd{laminar} (\aka \emph{overlap-free})
when no two sets from~$\mathcal{S}$ overlap.
By extension, a set family~$\mathcal{S}$ is \defd{laminar}
if $\left(\bigcup\mathcal{S}, \mathcal{S}\right)$ is a laminar set system
(note this also requires that~$\emptyset\notin\mathcal{S}$, $\bigcup\mathcal{S}\in\mathcal{S}$, and~$\{a\}\in\mathcal{S}$ for every~$a\in\bigcup\mathcal{S}$).

A laminar family~$\mathcal{S}$ of subsets of~$U$
naturally defines a rooted tree
where the nodes are the sets from~$\mathcal{S}$,
the root is~$U$,
and the ancestor relation corresponds to set inclusion, \ie,  nodes corresponding to sets $S\subset S'\in\mathcal{S}$ are adjacent in the tree if there
  is no proper superset of $S$ in $\mathcal{S}$ that is a proper subset of $S'$.
We call this rooted tree
the \defd{laminar tree of~$U$ induced by $\mathcal{S}$}
(or \defd{laminar tree of~$(U,\mathcal{S})$}).\xspace
In this rooted tree,
the leaves are the singletons~$\{x\}$ for~$x\in U$,
which we identify with the elements themselves.
That is to say, $L(T)=U$.
Laminar trees have the property
that each inner node has at least two children due to the definition of set systems demanding each singleton being in the set system.
Observe also that the size of a laminar tree
is linearly bounded in the size of the universe.

\subsection{Logic and transductions}
We use \emph{relational structures} to model
both graphs and the various tree-like decompositions used in this paper.
In order to concisely model set systems
we use the more general notion of \emph{extended relational structures},
namely \emph{relational structure} extended with set predicate names.
Such structures also naturally arise as outputs of \emph{\MSO-transductions} defined below.
\smallbreak

Define an \defd{(extended) vocabulary} to be a set of symbols,
each being either a \emph{relation} name~$R$ with associated arity~\defd{$\arity(R)\in\mathbb{N}$},
or a \emph{set predicate} name~$P$ with associated arity~\defd{$\arity(P)\in\mathbb{N}$}.
Set predicate names are aimed to describe relations between sets,
\eg, one may have a unary set predicate for selecting finite sets of even size,
or a binary set predicate for selecting pairs of disjoint sets.
We use capital~$R$ or
names starting with a lowercase letter (\eg, $\edge$, $\ancestor$, $\tedge$) for relation names,
and capital~$P$ or uppercase names (\eg, $\SET$, $\C[2]$) for set predicate names.
To refer to an arbitrary symbol of undetermined kind, we use capital~$Q$.
A \defd{relational vocabulary} is an extended vocabulary in which every symbol is a relation name.

Let~$\Sigma$ be a vocabulary.
An \defd{extended relational structure over~$\Sigma$} (\defd{$\Sigma$-structure})
is a structure $\mathbb{A}=\langle U_{\mathbb{A}}, (Q_{\mathbb{A}})_{Q\in \Sigma}\rangle$
consisting on the one hand of a set~\defd{$U_\mathbb{A}$} called \defd{universe},
and on the other hand, for each symbol~$Q$ in~$\Sigma$,
an \defd{interpretation~$Q_{\mathbb{A}}$ of~$Q$},
which is a relation of arity~$\arity(Q)$
either over the universe if~$Q$ is a relation name,
or over the family of subsets of the universe if~$Q$ is a set predicate name.
When~$\Sigma$ is not extended, $\mathbb{A}$ is simply a \defd{relational structure}. 

Given a $\Sigma$-structure~$\mathbb{A}$
and,
for some vocabulary~$\Gamma$,
a $\Gamma$-structure~$\mathbb{B}$,
we write~\defd{$\mathbb{A}\sqsubseteq\mathbb{B}$}
if~$\Sigma\subseteq\Gamma$,
$U_{\mathbb{A}}\subseteq U_{\mathbb{B}}$
and for each symbol~$Q$ in~$\Sigma$,
$Q_{\mathbb{A}}=Q_{\mathbb{B}}$.%
\footnote{%
    We require equality here
    (in particular only elements or subsets of~$U_{\mathbb{A}}$ are related in~$Q_{\mathbb{B}}$).
    This differs from classical notions of inclusions of relational structures
    which typically require equality only on the restriction of the universe to~$U_{\mathbb{A}}$,
    \ie, $Q_{\mathbb{A}}=Q_{\mathbb{B}_{/U_{\mathbb{A}}}}$,
    \eg, in order to correspond to induced graphs.%
}
We write~\defd{$\mathbb{A}\sqcup\mathbb{B}$}
to denote the $(\Sigma\cup\Gamma)$-structure
consisting of the universe~$U_{\mathbb{A}}\cup U_{\mathbb{B}}$
and, for each symbol~$Q\in\Sigma\cup\Gamma$,
the interpretation $Q_{\mathbb{A}\sqcup\mathbb{B}}$
which is~$Q_{\mathbb{A}}$, $Q_{\mathbb{B}}$, or~$Q_{\mathbb{A}}\cup Q_{\mathbb{B}}$
according to whether~$Q$ belongs to~$\Sigma\setminus\Gamma$, to~$\Gamma\setminus\Sigma$, or to~$\Sigma\cap\Gamma$.%
\footnote{%
    We do not require~$\mathbb{A}$ and~$\mathbb{B}$ to be disjoint structures
    whence we may have~$\mathbb{A}\not\sqsubseteq\mathbb{A}\sqcup\mathbb{B}$.%
}

To describe properties of (extended) relational structures, we use \defd{monadic second order logic} (\defd{\MSO}) and refer for instance to
\cite{CourcelleE12,FunkMN22,Hlineny06,Strozecki11} for the definition of \MSO on extended relational structures such as matroids or set systems in general. This logic allows quantification both over single elements of the universe
and over subsets of the universe. We also use \defd{counting \MSO} (\defd{\CMSO}), which is the extension of~$\MSO$ with, for every positive integer~$p$, a
unary set predicate~$\C[p]$ that checks whether the size of a subset is divisible by~$p$ or not. We only use~$\C[2]$.
As usual, lowercase variables indicate first-order-quantified variables,
while uppercase variables indicate monadic-quantified variables.
For a formula~$\phi$, we write, \eg,~$\phi(x,y,X)$ to indicate that the variables~$x$, $y$, and~$X$ belong to the set of \defd{free variables of~$\phi$},
namely, the set of variables occurring in~$\phi$ that are not bound to a quantifier within~$\phi$.
A \defd{sentence} is a formula without free variables.
\smallbreak

We now fix some (extended) vocabularies that we will use.
\begin{description}[nosep]
    \item[Graphs]
        To model both graphs, unrooted trees, and directed graphs,
        we use the relational vocabulary~$\{\edge\}$
        where~$\edge$ is a relation name of arity~$2$.
        A (directed) graph~$G=(V,E)$ is modelled
        as the $\{\edge\}$-structure $\mathbb{G}$
        with universe~$U_{\mathbb{G}}=V$
        and interpretation~$\edge_{\mathbb{G}}=E$.
        In particular, if~$G$ is undirected, then $\edge_{\mathbb{G}}$ is symmetric.
    \item[Rooted trees]
        We use the relational vocabulary~$\{\ancestor\}$ to model rooted trees
        where $\ancestor$ is a relation name of arity~$2$.
        A rooted tree~$T$ is modelled
        as the $\{\ancestor\}$-structure $\mathbb{T}$
        with universe~$U_{\mathbb{T}}=V(T)$
        and the interpretation~$\ancestor_{\mathbb{T}}$
        being the set of pairs~$(u,v)$ such that~$u$ is an ancestor of $v$ in $T$.
        It is routine to define \FO-formulas over this vocabulary
        to express the binary relations parent, child, proper ancestor, (proper) descendant,
        as well as the unary relations leaf and root.%
    \item[Set systems]
        To model set systems, we use the extended vocabulary~$\{\SET\}$
        where $\SET$ is a set predicate name of arity~$1$.
        A set system~$S=(U,\mathcal{S})$ is thus naturally modelled
        as the $\{\SET\}$-structure $\mathbb{S}$
        with universe~$U_{\mathbb{S}}=U$
        and interpretation~$\SET_{\mathbb{S}}=\mathcal{S}$.
\end{description}

\paragraph*{Transductions}
Let~$\Sigma$ and~$\Gamma$ be two extended vocabularies.
A \defd{$\Sigma$-to-$\Gamma$ transduction}
is a set~$\tau$ of pairs formed by
a $\Sigma$-structure, call the \defd{input},
and a $\Gamma$-structure, called the \defd{output}.
We write \defd{$\mathbb{B}\in\tau(\mathbb{A})$} when~$(\mathbb{A},\mathbb{B})\in\tau$. We say that a class of $\Sigma$-structures $\mathcal{C}$ \emph{transduces} a class of $\Gamma$-structures $\mathcal{D}$ if there is a  $\Sigma$-to-$\Gamma$ transduction $\tau$ with $\mathcal{D}\subseteq \tau(\mathcal{C})$.
When for every pair~$(\mathbb{A},\mathbb{B})\in\tau$
we have~$\mathbb{A}\sqsubseteq\mathbb{B}$,
we call~$\tau$ an
\defd{overlay transduction}. Overlay transductions are used to augment a relational structure by adding additional relations while not altering existing relations in any way.
Some transductions
can be defined by means of \MSO- or \CMSO[2]-formulas. 
This leads to the notion of \defd{\MSO-} and \defd{\CMSO[2]-transductions}.
Following the presentation of~\cite{BojanczykGP21},
for~$\mathsf{L}$ denoting~\MSO or~\CMSO[2],
define an \defd{$\mathsf{L}$-transduction}
to be a transduction obtained
by composing a finite number of \defd{atomic $\mathsf{L}$-transductions}
of the following kinds.
\begin{description}[nosep]
    \item[Filtering]
        An overlay transduction
        specified by an $\mathsf{L}$-sentence~$\phi$ over the input vocabulary~$\Sigma$,
        which discards the inputs that do not satisfy~$\phi$
        and keeps the other unchanged.
        Hence, it defines a partial function
        (actually, a partial identity)
        from $\Sigma$-structures to $\Sigma$-structures.
    \item[Universe restriction]
        A transduction specified by an $\mathsf{L}$-formula~$\phi$ over the input vocabulary~$\Sigma$,
        with one free first-order variable,
        which restricts the universe to those elements that satisfy~$\phi$.
        The output vocabulary is~$\Sigma$
        and the interpretation of every relation (\resp every predicate) in the output structure
        is defined as the restriction of its interpretation in the input structure
        to those tuples of elements satisfying~$\phi$ (\resp tuples of sets of elements that satisfy~$\phi$).
        This defines a total function from $\Sigma$-structures to $\Sigma$-structures.
    \item[Interpretation]
        A transduction specified by a family~$(\phi_Q)_{Q\in\Gamma}$ over the input vocabulary~$\Sigma$
        where~$\Gamma$ is the output vocabulary and each~$\phi_Q$ has $\arity(Q)$ free variables
        which are first-order if~$Q$ is a relation name and monadic if it is a set predicate name.
        The transduction outputs the~$\Gamma$-structure
        that has the same universe as the input structure
        and in which each relation or predicate~$Q$ is interpreted
        as those set of tuples that satisfy~$\phi_Q$.
        This defines a total function from $\Sigma$-structures to $\Gamma$-structures.
    \item[Copying]
        An overlay transduction parametrized by a positive integer~$k$ that
        adds $k$ copies of each element to the universe.
        The output vocabulary consists in the input vocabulary~$\Sigma$
        extended with~$k$ binary relational symbols~$(\copyV[i])_{i\in[k]}$
        interpreted as pairs of elements~$(x,y)$
        saying that ``$y$ is the $i$-th copy of~$x$''.
        The interpretation of the relations (\resp predicates) of the input structure
        are preserved, on original elements.
        This defines a total function from~$\Sigma$-structures to~$\Gamma$-structures,
        where~$\Gamma=\Sigma\cup\{\copyV[i]\mid i\in[k]\}$.
    \item[Colouring]
        An overlay transduction that adds a new unary relation~$\colorV\notin\Sigma$ to the structure.
        Any possible interpretation yields an output. 
        The interpretation of the relations (\resp predicates) of the input structure
        are preserved.
        Hence, it defines a total (non-functional) relation
        from $\Sigma$-structures to~$\Gamma$-structures
        where~$\Gamma=\Sigma\cup\{\colorV\}$.
        Alternatively, we can think of the interpretation of $\colorV$ being chosen non-deterministically.
        Note that by composing several transductions of this type any given element may receive multiple colours.
\end{description}
We say an $\mathsf{L}$-transduction is \defd{deterministic} if it does not use colouring,
it is \defd{non-deterministic} otherwise.
By definition, deterministic $\mathsf{L}$-transductions define partial functions.
    It is known that every $\CMSO$-transduction
    can be put in a standard form in which each of the above atomic transductions occurs once
    in a fixed order, namely:
    first colouring, then filtering, copying, interpreting the output relations, and finally restricting the universe~\cite{CourcelleE12}. 

\section{Transducing the laminar-tree from a laminar set system}
\label{sec:transduce laminar tree}
In this section,
we present an overlay \CMSO[2]-transduction
that takes as input a laminar set system
and outputs the laminar tree it induces.

\begin{thm}
    \label{thm:transduce laminar tree}
    Let~$\Sigma$ be an extended vocabulary,
    including a unary set predicate name~$\SET$
    and not including the binary relational symbol~$\ancestor$.
    There exists a non-deterministic overlay \CMSO[2]-transduction~$\tau$
    such that,
    for each laminar set system~$(U,\mathcal{S})$
    represented as the $\{\SET\}$-structure~$\mathbb{S}$
    and inducing the laminar tree~$T$ with ${L(T)=U}$,
    and for each $\Sigma$-structure~$\mathbb{A}$
    with~$\mathbb{S}\sqsubseteq\mathbb{A}$,
    $\tau(\mathbb{A})$ is non-empty
    and every output in~$\tau(\mathbb{A})$
    is equal to $\mathbb{A}\sqcup\mathbb{T}$ for some $\{\ancestor\}$-structure $\mathbb{T}$ representing $T$.
\end{thm}
Note that the theorem is formulated such that it can be applied to a structure of which the laminar set system is a part of (e.g. a graph and the laminar set system of strong modules). The original structure is maintained by the transduction and the laminar tree structure is added. This is necessary in order to use \autoref{thm:transduce laminar tree} as a black-box in more complex transductions (see application sections below).

Since the sets from~$\mathcal{S}$
are precisely the sets of leaves of the subtrees of~$T$,
there is an \MSO-transduction
which is the inverse of the above \CMSO[2]-transduction;
that is to say, given an $\{\ancestor\}$-structure~$\mathbb{T}$ representing the laminar tree, it
outputs the original set system~$\mathbb{S}$.
Namely,
\begin{enumerate}[nosep]
    \item An interpretation $\Phi_{\SET}(S)$
    that is true for a set $S$ when there exists an element $a$ such that an element $x$ is in $S$ if and only if $a$ is an ancestor of $x$.
    \item The filtering $\phi(x)$ keeping leaves only.
\end{enumerate}
\medbreak

We prove~\autoref{thm:transduce laminar tree} in \autoref{ssec:transduce laminar tree},
using the key tools developed in~\autoref{ssec:repr},
that allow us to represent each inner node of~$T$ with a pair of leaves from its subtree,
while keeping, for each leaf, the number of inner nodes it represents bounded.

\subsection{Inner node representatives}
\label{ssec:repr}
In this section, the root of any rooted tree is always an inner node (unless the tree is a unique node).
We fix a rooted tree~$T$,\xspace%
in which every inner node has at least two children (a necessary assumption that is satisfied by laminar trees),
and we let~$V$ denote the set of its nodes ($V=V(T)$)
and~$L\subseteq V$ denote the subset of its leaves ($L=L(T)$).
\smallbreak

Let~$S\subseteq V\setminus L$ be a set of inner nodes,
and let~$(\pi,\sigma)$ be a pair of injective mappings from~$S$ to~$L$.
We say that the pair~$(\pi,\sigma)$ \emph{identifies~$S$}
if for each~$s\in S$, $s$ is the least common ancestor of~$\pi(s)$ and~$\sigma(s)$.
For~$s\in S$ and~$x\in V$,
we say that~$x$ is \emph{$s$-requested in~$(\pi,\sigma)$}
if~$x$ lies on the path from~$\pi(s)$ to~$\sigma(s)$
(namely, on either of the paths from~$\pi(s)$ to~$s$ and from~$\sigma(s)$ to~$s$).
We say that~$x$ is \emph{requested in~$(\pi,\sigma)$} if it is $s$-requested for some~$s$.
The pair~$(\pi,\sigma)$ has \emph{unique request}
if every node~$x$ of~$T$ is requested at most once in~$(\pi,\sigma)$,
\ie, there exists at most one~$s\in S$ such that~$x$ is $s$-requested in~$(\pi,\sigma)$. Note that if $(\pi,\sigma)$ has unique request then the paths from $\pi(s)$ to $\sigma(s)$ are pairwise disjoint for all the $s\in S$.
We now state a few basic observations.
\begin{rem}
    \label{rk:identification basic observations}
    Let~$(\pi,\sigma)$ identifying some subset~$S$ of~$V\setminus L$.
    \begin{enumerate}
        \item\label{rk:identification basic observations:commute}
            The reversed pair $(\sigma,\pi)$ also identifies~$S$
            and has unique request whenever~$(\pi,\sigma)$ does.
        \item\label{rk:identification basic observations:sub}
            If~$S'\subset S$, then $(\pi_{|S'},\sigma_{|S'})$ identifies~$S'$,
            and has unique request if~$(\pi, \sigma)$ does.
        \item\label{rk:identification basic observations:self-request}
            For each~$s\in S$, $s$ is $s$-requested in~$(\pi,\sigma)$.
        \item\label{rk:identification basic observations:disjoint}
            If $(\sigma,\pi)$ has unique request,
            then~$\pi(S)$ and~$\sigma(S)$ are disjoint subsets of~$L$.
    \end{enumerate}
\end{rem}

\begin{lem}
    \label{lem:a reprensants}
    Let~$(\pi,\sigma)$ identifying some subset~$S$ of~$V\setminus L$ with unique request.
    Then for each~$a\in\pi(S)$, 
    the node~$\pi^{-1}(a)$ 
    is the least ancestor of~$a$ 
    which belongs to~$S$.
\end{lem}
\begin{proof}
    Let~$a\in\pi(S)$ and let~$s=\pi^{-1}(a)$.
    By definition, $s$ is an ancestor of~$a$ which belongs to~$S$.
    Let~$y$ be the least ancestor of $a$ that is contained in $S$. As $s$ is an ancestor of $a$ belonging to $S$, $y$ must be a descendant of $s$. Hence, $y$ is $s$-requested in~$(\pi,\sigma)$. Additionally, by \autoref{rk:identification basic observations:self-request} of \autoref{rk:identification basic observations},
    $y$ is $y$-requested in~$(\pi,\sigma)$.
    Since~$(\pi,\sigma)$ has unique request, $y=s$.
    Thus, $s$ is the least ancestor of~$a$ which belongs to~$S$.
\end{proof}

Notice that,
by \autoref{rk:identification basic observations:commute} of \autoref{rk:identification basic observations},
a similar result holds for each~$b\in\sigma(S)$.
It follows that the sets~$\pi(S)$ and~$\sigma(S)$
characterize~$(\pi,\sigma)$.
\begin{lem}
    Let~$S\subseteq V\setminus L$,
    and $(\pi,\sigma)$ and $(\pi',\sigma')$
    be two pairs of injections from~$S$ to~$L$
    identifying~$S$ with unique request.
    If~$\pi(S)=\pi'(S)$ and~$\sigma(S)=\sigma'(S)$,
    then $\pi=\pi'$ and~$\sigma=\sigma'$.
\end{lem}
\begin{proof}
    Let~$s\in S$, $a=\pi(s)$, and~$s'=\pi'^{-1}(a)$.
    Both~$s$ and~$s'$ are the least ancestor of~$a$ which belongs to~$S$,
    hence~$s=s'$ and~$\pi'(s)=a$.
    Thus~$\pi=\pi'$.
    By \autoref{rk:identification basic observations:commute} of \autoref{rk:identification basic observations},
    we also obtain that~$\sigma=\sigma'$.
\end{proof}

Let~$A$ and~$B$ be two subsets of~$L$ that are disjoint and of same cardinality.
We call such a pair~$(A, B)$ a \emph{bi-colouring}.
We say that~$(A,B)$ \emph{identifies~$S$}
if there exists a pair~$(\pi,\sigma)$ identifying~$S$ with unique request
such that~$\pi(S)=A$ and~$\sigma(S)=B$.
By the previous lemma, for a fixed set $S\subseteq V\setminus L$  and a fixed bi-colouring $(A,B)$ of $L$,
the pair $(\pi,\sigma)$ identifying $S$ with $\pi(S)=A$ and $\sigma(S)=B$ is unique when it exists.
We will also prove that~$S$ is actually uniquely determined from~$(A,B)$.
Before that, we state the following technical lemma.
\begin{lem}
    \label{lem:node cases in identification}
    Let~$(A,B)$ identify some subset~$S$ of~$V\setminus L$
    through a pair~$(\pi,\sigma)$ of injections having unique request.
    Then,
    for each inner node~$x$,
    exactly one of the three following cases holds:
    \begin{enumerate}
        \item
            $x\notin S$,
            $x$ is not requested in~$(\pi,\sigma)$,
            and for each child~$c$ of~$x$,
            $\card{A\cap V(T_c)}=\card{B\cap V(T_c)}$;
        \item
            $x\notin S$,
            $x$ is requested in~$(\pi,\sigma)$,
            and there exists one leaf~$z\in (A\cup B)\cap V(T_x)$
            such that, for each child~$c$ of~$x$,
            $\card{(A\setminus\{z\})\cap V(T_c)}=\card{(B\setminus\{z\})\cap V(T_c)}$;
        \item
            $x\in S$,
            $x$ is requested in~$(\pi,\sigma)$,
            and there exists two leaves
            $a\in A\cap V(T_x)$ and $b\in B\cap V(T_x)$
            such that, for each child~$c$ of~$x$,
            $\card{(A\setminus\{a\})\cap V(T_c)}=\card{(B\setminus\{b\}\cap V(T_c))}$
            and~$\{a,b\}\nsubseteq V(T_c)$.
    \end{enumerate}
    In particular, $x\in S$ \ifof there exists 
            $a\in A\cap V(T_x)$ and $b\in B\cap V(T_x)$
            such that, for each child~$c$ of~$x$,
            $\card{(A\setminus\{a\})\cap V(T_c)}=\card{(B\setminus\{b\}\cap V(T_c))}$
            and~$\{a,b\}\nsubseteq V(T_c)$.
\end{lem}
For an illustration of the different cases see \autoref{fig:colouringBinaryCotree} where $x_1$ is an example of a node satisfying the first case, $x_2$ an example for the second case and $x_3$ for the third.
\begin{proof}
    Let~$x$ be an inner node.
    We consider the set~$S'=S\setminus(V(T_x)\setminus\{x\})$ of all nodes which are not proper descendants of $x$
    and the restrictions~$\pi'$ and~$\sigma'$ of, respectively,~$\pi$ and~$\sigma$ to~$S'$.
    By \autoref{rk:identification basic observations:sub} of \autoref{rk:identification basic observations}
    $(\pi',\sigma')$ identify~$S'$ with unique request.
    We denote~$A'=\pi'(S')$ and~$B'=\sigma'(S')$, thus~$(A',B')$ identify~$S'$.
    Let $A'_x=A'\cap V(T_x)$ and $B'_x=B'\cap V(T_x)$ be the sets of images under $\pi'$, $\sigma'$ restricted to $A'$ and $B'$, respectively, contained in $T_x$.
    In case a leaf~$a$ is an element of~$A'_x$,
    the element~$s_a=\pi^{-1}(a)$ is an ancestor of~$x$
    because it is an ancestor of~$a\in V(T_x)$
    and belongs to~$S'$ whence not to~$V(T_x)\setminus\{x\}$.
    Therefore,~$x$ is~$s_a$-requested.
    As~$(\pi',\sigma')$ has unique request,
    there exists at most one element in~$A'_x$.
    Similarly, $B'_x$ has size at most~$1$.
    Moreover, $A'_x\cap B'_x=\emptyset$
    by~\autoref{rk:identification basic observations:disjoint} of \autoref{rk:identification basic observations}.
    We thus we have three cases:
    \begin{description}
        \item[Case~$A'_x\cup B'_x=\emptyset$]
            Then there is no~$s\in S'$ such that~$\pi(s)\in V(T_x)$ or~$\sigma(s)\in V(T_x)$.
            Hence, $x$ is not requested in~$(\pi,\sigma)$, in particular, $x\notin S$.

            Let~$c$ be a child of~$x$.
            If~$\card{A\cap V(T_{c})}\neq\card{B\cap V(T_c)}$,
            then there exists~$a\in(A\cup B)\cap V(T_{c})$
            such that,
            assuming without loss of generality that~$a\in A$ and denoting~$s_a=\pi^{-1}(a)$,
            $\sigma(s_a)\notin V(T_{c})$.
            Hence, $s_a$ is a proper ancestor of~$c$,
            thus, equivalently, an ancestor of~$x$,
            implying that~$x$ is $s_a$-requested in~$(\pi,\sigma)$.
            This contradicts the above argument.
            So, $\card{A\cap V(T_c)}=\card{B\cap V(T_c)}$ for each child~$c$ of~$x$.
        \item[Case~$A'_x\cup B'_x=\{a\}$]
            Assume, without loss of generality, that~$a\in A$,
            and denote $s_a=\pi^{-1}(a)$ and~$b=\sigma(s_a)$.
            We have that~$s_a$ is a proper ancestor of~$x$,
            since it is the least common ancestor of~$a\in V(T_x)$ and~$b\notin V(T_x)$.
            Thus, $x$ is $s_a$-requested and~$s_a\neq x$ so~$x\notin S$.

            Let~$c$ be a child of~$x$.
            If~$\card{A\cap V(T_{c})}\neq\card{B\cap V(T_c)}$,
            then it means that there exists $z\in(A\cup B)\cap V(T_{c})$,
            such that
            either~$z\in A$ and $s_z=\pi^{-1}(z)\notin V(T_{c})$,
            or~$z\in B$ and~$s_z=\sigma^{-1}(z)\notin V(T_{c})$.
            In both cases, $s_z$ is a proper ancestor of~$c$,
            whence an ancestor of~$x$.
            Thus, $x$ is $s_z$-requested in~$(\pi,\sigma)$.
            However, $x$ is $s_a$-requested in~$(\pi,\sigma)$ which has unique request,
            hence~$s_z=s_a$.
            If~$z\in B$, it follows that~$z=b$, which contradicts the fact~$b\notin V(T_x)$.
            Hence, $z\in A$ and thus, $z=\pi(s_a)=a$.
            On the other hand, if $\card{A\cap V(T_{c})}=\card{B\cap V(T_c)}$ then $a\notin V(T_c)$. If $a\in V(T_c)$ then $c$ would be $s_a$-requested and additionally (as $\card{(A\setminus \{a\}\cap V(T_{c})}<\card{B\cap V(T_c)}$) there has to be $z\in T_c$,  $z\in B$ for which $s_z=\sigma(z)^{-1}$ is an ancestor of $c$. In particular, $s_a\neq s_z$ as $c$ is the least common ancestor of $a$ and $z$, but $s_a$ is a proper ancestor of $c$. Hence, $c$ is requested by both $s_a$ and $s_z$ contradicting the assumption that $(\pi,\sigma)$ has unique requests.
            Therefore, $\card{(A\setminus\{a\})\cap V(T_{c})}=\card{B\cap V(T_c)}$
            for each child~$c$ of~$x$. 
        \item[Case~$A'_x=\{a\}$ and~$B'_x=\{b\}$]
            Let~$s_a=\pi^{-1}(a)$ and~$s_b=\sigma^{-1}(b)$.
            The node $x$ is $s_a$-requested and $s_b$-requested in $(\pi,\sigma)$,
            so, by the unique request property, $s_a=s_b$.
            Since~$s_a$ is the least common ancestor of~$a$ and~$b$,
            it belongs to~$V(T_x)$
            whence~$s_a=x$ implying~$x\in S$.

            Let~$c$ be a child of~$x$.
            If~$\card{A\cap V(T_c)}\neq\card{B\cap V(T_c)}$,
            then there exists~$z\in (A\cup B)\cap V(T_{c})$
            such that either~$z\in A$ and~$s_z=\pi^{-1}(z)\notin V(T_{c})$,
            or~$z\in B$ and~$s_z=\sigma^{-1}(z)\notin V(T_{c})$.
            In both cases, $s_z$ is a proper ancestor of~$c$,
            whence an ancestor of~$x$,
            and thus $x$ is $s_z$-requested in~$(\pi,\sigma)$.
            However, $x$ is $x$-requested in~$(\pi,\sigma)$ which has unique request,
            hence~$s_z=x$ and thus~$z\in\{a,b\}$. On the other hand, similarly to the previous case, if~$\card{A\cap V(T_c)}=\card{B\cap V(T_c)}$,
            then neither $a$ nor $b$ can be contained in $V(T_c)$.
            Therefore, $\card{(A\setminus\{a\})\cap V(T_{c})}=\card{(B\setminus\{b\})\cap V(T_c)}$
            for each child~$c$ of~$x$.
            Because the least common ancestor of~$a$ and~$b$ is~$x$,
            there is no child~$c$ of~$x$ containing both~$a$ and~$b$ as leaves,
            \ie, $\{a,b\}\nsubseteq V(T_c)$.
    \end{description}
    This concludes the proof of the statement.
\end{proof}

The next lemma shows that for each bi-colouring~$(A, B)$,
there exists at most one set~$S$ of inner nodes
identified by~$(A, B)$.
\begin{lem}
    Let~$(A,B)$ be two disjoint subsets of~$L$
    and let~$S$ and~$S'$ be two subsets of~$V\setminus L$.
    If~$(A,B)$ identify both~$S$ and~$S'$, then~$S=S'$.
\end{lem}
\begin{proof}
    We proceed by contradiction and thus assume~$S\neq S'$.
    Let~$s\in S\setminus S'$.
    Since $s\in S$ by \autoref{lem:node cases in identification}(3),
    there are two leaves $a\in A\cap V(T_c)$ and $b\in B\cap V(T_c)$ such that, for each child~$c$ of~$x$,
    $\card{(A\setminus\{a\})\cap V(T_c)}=\card{(B\setminus\{b\}\cap V(T_c))}$
    and~$\{a,b\}\nsubseteq V(T_c)$. Let $c_a$ and $c_b$ be the two children of $s$ with $a\in V(T_{c_a})$ and $b\in V(T_{c_b})$. Since $\{a,b\}\nsubseteq V(T_c)$ for every child $c$ of $s$, we get $c_a\neq c_b$. In particular, $\card{A\cap V(T_{c_a})}=\card{(B\cap V(T_{c_a}))}+1$,  $\card{A\cap V(T_{c_b})}+1=\card{(B\cap V(T_{c_b}))}$ and $\card{A\cap V(T_{c})}=\card{(B\cap V(T_{c}))}$ for every  other children $c$ of $s$. Since $A$ and $B$ are disjoint by \autoref{rk:identification basic observations:disjoint} of \autoref{rk:identification basic observations}, the number of children~$c$ of~$s$
    for which the set~$(A\cup B)\cap V(T_{c})$ has odd size is~$2$.
    
    On the other hand, since~$s\notin S'$, \autoref{lem:node cases in identification}(1) and (2) imply, using a similar argument to above, that there are no children $c$ of $s$ for which $(A\cup B)\cap V(T_{c})$ has odd size (in case (1)) or there is~$1$ child $c$ of $s$ for which $(A\cup B)\cap V(T_{c})$ has odd size. This contradicts our conclusion from the previous paragraph that $s$ should have~$2$ such children.    
\end{proof}

When~$(A,B)$ identifies a set~$S$,
we call \emph{$A$-representative} (\resp \emph{$B$-representative}) of~$s\in S$
the leaf~$\pi(s)\in A$ (\resp $\sigma(s)\in B$),
where~$(\pi,\sigma)$ witnessing that~$(A,B)$ identifies~$S$.
An example of bi-colouring identifying a subset of inner nodes
is given in \autoref{fig:colouringBinaryCotree}.
\begin{figure}
    \includegraphics[width=\textwidth]{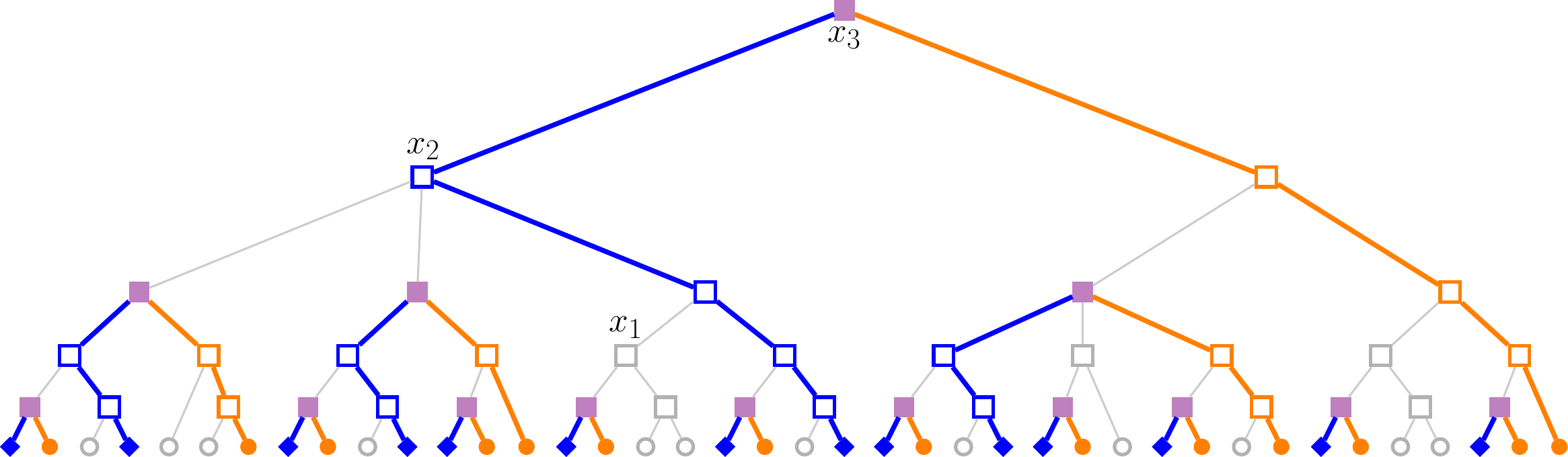}%
    \caption{%
        Illustration of a bi-colouring~$(A,B)$
        identifying some set $S\subseteq V\setminus L$
        in a binary tree~$T$.
        Leaves from~$A$ are blue diamonds,
        leaves from~$B$ are  orange filled circles,
        and inner nodes from~$S$ are filled purple squares.
        Furthermore, the two paths
        connecting an inner node~$s\in S$
        to its~$A$- and~$B$-representatives
        are coloured blue and orange, respectively.
        Nodes labelled $x_1,x_2,x_3$ constitute examples illustrating the different cases in \autoref{lem:node cases in identification}.%
    }%
    \label{fig:colouringBinaryCotree}
\end{figure}
\bigbreak

Not every set of inner nodes
has a bi-colouring identifying it.
To ensure that such a pair exists,
we consider \emph{thin sets of inner nodes}. While thin sets always have bi-colourings identifying them, it is also guaranteed that the set of inner nodes can be partitioned into only 4 thin sets.
A subset~$X\subseteq V\setminus L$ is \emph{thin}
when, for each~$x\in X$ not being the root,
on the one hand, the parent~$p_x$ of~$x$ does not belong to~$X$,
and on the other hand, $x$ admits at least one sibling (including possible leaves) that does not belong to~$X$.
Having a thin set~$X$ allows to find branches avoiding it.
\begin{lem}
    \label{lem:thin => path avoiding}
    Let~$X\subseteq V\setminus L$ and~$s\in V\setminus X$.
    If~$X$ is thin, then there exists a leaf~$t\in V(T_s)$
    such that the path from~$t$ to~$s$ avoids~$X$
    (\ie, none of the nodes along this path belong to~$X$).
\end{lem}
\begin{proof}
    If~$T_s$ has height~$0$,
    then~$s$ is a leaf
    and taking~$t=s$ trivially gives the expected path.
    Otherwise,
    $s$ is an inner node and,
    because~$X$ is thin,
    $s$ has at least one child~$c_s$ not belonging to~$X$.
    By induction, there is a path from some leaf~$t\in V(T_{c_s})$ to~$c_s$ avoiding~$X$ and,
    since~$s\notin X$,
    this path could be extended into a path from~$t$ to~$s$ avoiding~$X$.
\end{proof}

The following lemma allows to identify every thin set.
\begin{lem}
    \label{lem:thin => identifiable}
    If~$X$ is a thin set,
    then there exists a pair~$(\pi,\sigma)$
    of injections from~$X$ to~$L$
    that has unique request and that identifies~$X$.
\end{lem}
\begin{proof}
    We proceed by induction on the size of~$X$.
    If~$X=\emptyset$, the result is trivial.
    Let~$n\in\mathbb{N}$ and suppose that for every thin set of size~$n$
    there exists a pair of injections identifying it with unique request.
    Let~$X$ be a thin set of size~$n+1$,
    and let~$s\in X$ be of minimal depth.
    Clearly, $X\setminus\{s\}$ is thin
    and thus there exists, by induction, a pair~$(\pi,\sigma)$
    of injections from~$X\setminus\{s\}$ to~$L$ identifying~$X\setminus\{s\}$ with unique request.
    Since~$X$ is thin and~$s\in X$,
    we can find two distinct children~$c_a$ and~$c_b$ of~$s$
    not belonging to~$X$.%
    \footnote{\label{fn:out-degree at least 2}Remember that every inner node of~$T$ has at least two children (including possible leaves).}
    Then, by \autoref{lem:thin => path avoiding},
    there exists a leaf~$a\in V(T_{c_a})$ (\resp a leaf~$b\in V(T_{c_b})$)
    such that the path from~$a$ to~$c_a$ (\resp from~$b$ to~$c_b$)
    avoids~$X$.
    In particular, for each node~$y$ along these paths,
    since~$y$ has no ancestor that belongs to~$X$ but~$s$,
    $y$ is not requested in~$(\pi,\sigma)$.
    Hence,
    extending~$\pi$ (\resp~$\sigma$) in such a way that,
    besides mapping each~$x\in X\setminus\{s\}$ to~$\pi(x)$ (\resp to~$\sigma(x)$),
    it maps~$s$ to~$a$ (\resp to~$b$),
    we obtain a pair~$(\hat{\pi},\hat{\sigma})$
    of injections from~$X$ to~$L$
    that identifies~$X$ with unique request.
\end{proof}

A family $F=(A_1,B_1),\ldots,(A_n,B_n)$ of bi-colourings
\emph{identifies} a set~$S\subseteq V\setminus L$,
if there exists a partition~$(S_1,\ldots,S_n)$ of~$S$
such that, for each~$i\in[n]$, $(A_i,B_i)$ identifies~$S_i$.
Whenever $S=V\setminus L$ we say that~$F$ \emph{identifies~$T$}.
Note that while $A_i$ and $B_i$ must be disjoint, for $i\not= j$ it is possible (and sometimes even necessary) that $A_i$ and $A_j$ or $A_i$ and $B_j$ are not disjoint (see \eg~\autoref{fig:partitionThinSets}).
A collection of subsets of~$V\setminus L$ is \emph{thin}
if each of its subsets is thin.
We now show that there exists a thin $4$-partition.
\begin{lem}
    \label{lem:thin 4-partition}
    There exists a thin $4$-partition of~$V\setminus L$.
\end{lem}
\begin{proof}
    We build such a thin $4$-partition as follows.
    First, consider the partition~$(D_e,D_o)$ of~$V\setminus L$
    in which~$D_e$ (\resp $D_o=V\setminus(D_e\cup L)$) is the set of all inner nodes of even (\resp odd) depth.
    Second, arbitrarily fix one child~$c_x$ of~$x$ for each inner node~$x$,
    and consider the set~$C=\{ c_x \mid x\in V\setminus L \}\setminus L$,
    inducing a partition~$(C,\overline{C})$ of~$V\setminus L$ where~$\overline{C}=V\setminus(L\cup C)$.
    By refining these two bi-partitions, we obtain a $4$-partition which is thin by construction.%
    \footref{fn:out-degree at least 2}
\end{proof}

Hence, four bi-colourings are enough to identify~$T$. For an example see \autoref{fig:partitionThinSets}.
\begin{cor}
    \label{cor:4 bi-colouring}
    There exists a family of four bi-colourings identifying~$T$.%
\end{cor}
\begin{proof}
    The result immediately follows from \autoref{lem:thin => identifiable} and \autoref{lem:thin 4-partition}.%
\end{proof}

\begin{figure}
    \includegraphics[width=\textwidth]{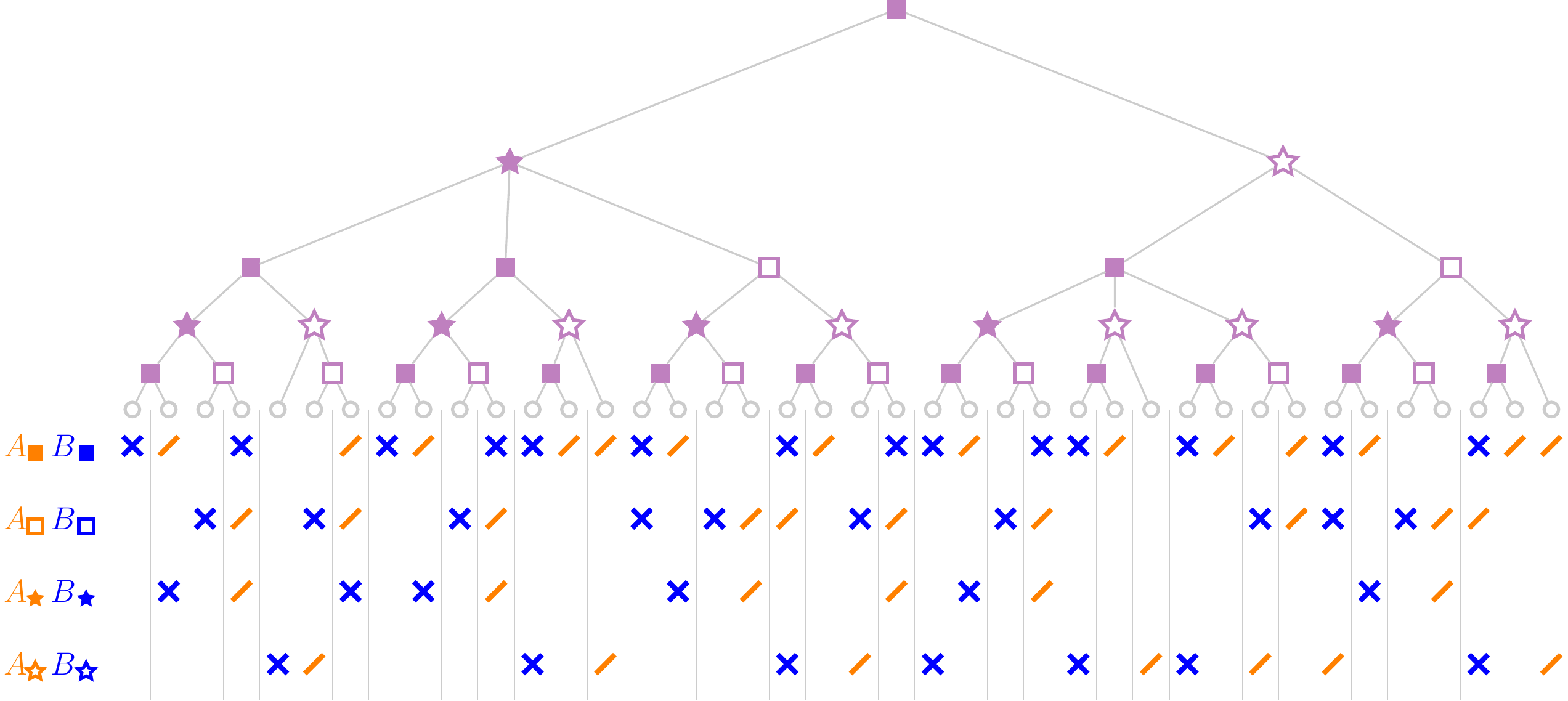}%
    \caption{%
        Illustration of a partition of the inner nodes into 4 thin sets indicated by different shapes/fillings.
        The four colourings identifying each thin set are indicated in the table below the leaves where an orange slash indicates the leaf being in the corresponding set $A$ and a blue cross indicates the leaf being in $B$}.%
    \label{fig:partitionThinSets}
\end{figure}
\bigbreak

\subsection{The transduction}
\label{ssec:transduce laminar tree}
The goal of this section is to prove \autoref{thm:transduce laminar tree},
that is, to design a \CMSO[2]-transduction
that produces the laminar tree induced by an input laminar set system.
For this section, we fix an arbitrary laminar set system~$(U,\mathcal{S})$,
represented by a $\{\SET\}$-structure~$\mathbb{S}$.
Before proving the theorem,
we make the following basic observation.
On $\mathbb{S}$,
we can define two \MSO-formulas~$\descendant(X,Y)$ and~$\child(X, Y)$
expressing that in the laminar tree induced by~$\mathcal{S}$,
$X$ and~$Y$ are nodes and~$X$ is a descendant or a child of~$Y$, respectively:
\begin{align*}
    \descendant(X, Y)&:=
    \SET(X)\land\SET(Y)\land X\subseteq Y;
    \\*
    \child(X, Y)&:=
    \descendant(X, Y)\land X\neq Y\land
    \forall Z\:
    \Big(\big(
    \descendant(Z, Y)
    \land
    Z\neq Y\big)
    \implies
    \descendant(Z, X)
    \Big).
\end{align*}

The key point of our construction consists in defining
a \CMSO[2]-formula
$\repr_{A,B}(a, X)$ which,
assuming a bi-colouring~$(A, B)$ of the universe
modelled as disjoint unary relations
and
identifying a subset~$S$ of inner nodes of~$T$,
is satisfied exactly when~$X\in S$ and~$a$ is its $A$-representative.
\begin{lem}
    \label{lem:repr MSO}
    Let~$(A,B)$ be a bi-colouring of~$L(T)$
    identifying a subset~$S$ of inner nodes of~$T$.
    There exists a \CMSO[2]-formula~$\repr_{A,B}(a, X)$
    that is satisfied exactly
    when~$X$ is an inner node of~$T$
    that belongs to~$S$ and is $A$-represented by~$a$.
\end{lem}
\begin{proof}
    First, we define a formula~$\phi_{A,B}(X)$
    that, under the above assumption,
    is satisfied exactly when~$X$ belongs to~$S$.
    According to~\autoref{lem:node cases in identification},
    this happens \ifof
    $X$ is a node of~$T$
    (\ie,~$\SET(X)$ is satisfied)
    and there exists~$a\in X\cap A$ and~$b\in X\cap B$
    such that for each child~$Z$ of~$X$,
    $\{a,b\}\nsubseteq Z$
    and
    the set~$(Z\setminus\{a,b\})\cap(A\cup B)$ has even size.
    This property is easily expressed in \CMSO[2],
    using the \MSO-formula $\child(X, Y)$ defined previously,
    as well as the predicate $\SET$:
    \begin{align*}
        \phi_{A,B}(X):=\:
        &
        \SET(X)
        \land
        \exists a\exists b\Big[a\in (X\cap A)\land b\in (X\cap B)\land
        \\*&\qquad
        \forall Z\:
        \Big(
        \child(Z, X) \implies \Big(
        \{a,b\}\nsubseteq Z
        \land
        \C[2]\big((Z\setminus\{a,b\})\cap(A\cup B)\big)
        \Big)\Big)\Big].
    \end{align*}

    Now, we can easily define~$\repr_{A,B}(a, X)$
    based on \autoref{lem:a reprensants}:
    \begin{align*}
        \repr_{A,B}(a, X):=\:&
        \phi_{A,B}(X) \land a\in (X\cap A)
        \land
        \forall Z\subsetneq X\:
        \Big(
        a\in Z
        \implies \neg \phi_{A,B}(Z)
        \Big).
    \end{align*}
    This concludes the proof.
\end{proof}

We are now ready to prove the theorem.
\begin{proof}[Proof of \autoref{thm:transduce laminar tree}]
    The \CMSO[2]-transduction is obtained by composing the following atomic \CMSO[2]-transductions.
    The transduction makes use of the formulas~$\repr_{A, B}(a, X)$
    given by \autoref{lem:repr MSO}.
    \begin{enumerate}
        \item Guess a family of four bi-colourings~$(A_i, B_i)_{i\in[4]}$
            identifying~$T$
            (which exists by \autoref{cor:4 bi-colouring}).
        \item Copy the input graph four times,
            thus introducing four binary relations~$(\copyV[i])_{i\in[4]}$
            where $\copyV[i](x, y)$ indicates that~$x$ is the $i$-th copy
            of the original element~$y$.
        \item Filter the universe
            keeping only
            the original elements
            as well as
            the $i$-th copy of each vertex~$a$
            for which there exists~$X$
            such that~$\repr_{A_i, B_i}(a, X)$.
        \item Use an interpretation adding relation $\ancestor$ and keeping every other relation in $\Sigma$ unchanged outputting $\Sigma\cup \{\ancestor\}$-structures as follows. Define the relation~$\ancestor(x, y)$
            so that it is satisfied exactly when
            there exist~$x'$, $X$, $i$, and~$Y$
            such that,
            on the one hand
            $\descendant(Y, X)$
            and
            $\copyV[i](x,x')\land\repr_{A_i, B_i}(x', X)$,
            and,
            on the other hand,
            either~$y$ is an original element and~$Y=\{y\}$,
            or there exists~$y'$ and~$j$ such that
            $\copyV[j](y, y')\land\repr_{A_j, B_j}(y', Y)$.%
            \item Use a filtering transduction keeping only those outputs $\mathbb{B}$ from the previous steps that satisfy that for every set $X$ in $\SET_{\mathbb{B}}$ there is an element $t$ in $U_{\mathbb{B}}$ such that the set of descendants of $t$ is equal to $X$ and for every element $t$ of $U_{\mathbb{B}}$ the set of descendants of $t$ is in $\SET_{\mathbb{B}}$ and hence verifying that $\ancestor_{\mathbb{B}}$ defines a laminar tree of $(U,\mathcal{S})$. Clearly, this can be expressed using a \MSO-sentence and hence outputs are of the form $\mathbb{A}\sqcup \mathbb{T}$ for some $\{\ancestor\}$-structure $\mathbb{T}$ representing the laminar tree $T$ of $(U,\mathcal{S})$. 
    \end{enumerate}
\end{proof}

\section{Transducing modular decompositions from graphs}
\label{sec:transduce modular decomposition}
The set of modules of a directed graph is a specific example of a particular type of set system, a ``weakly-partitive set system'' (and a ``partitive set
system'' in the case of an undirected graph). In this section, we first give a general \CMSO[2]-transduction to obtain the canonical tree-like decomposition of
a weakly-partitive set system from the set system itself. We then show how to obtain the modular decomposition of a graph via a \CMSO[2]-transduction as an
application. To capture the structure of weakly-partitive set systems we translate the structural theorem (\autoref{thm:weakly-partitive tree}) in the
  form given by Rao \cite{raoThesis}  into a transduction using \autoref{thm:transduce laminar tree}. Since weakly partitive set systems can be reconstructed by
  an \MSO-transduction from their canonical tree-like decompositions, we obtain as a corollary that weakly-partitive set systems are equivalent \emph{w.r.t.} \CMSO[2]-transductions to
  their canonical tree-like decompositions. It is worth mentioning that Courcelle provides similar transductions for the modular
  decomposition of directed graphs, but using order-invariant \MSO-formulas \cite{Courcelle06}, and the Tutte decomposition of 2-connected graphs, but using incidence
  relation \cite{Courcelle99}. Corollaries of our result is that we can compute such canonical decompositions using only \CMSO[2]-formulas and on the basis of the edge
  relation. This answers questions raised in \cite{Courcelle99,Courcelle06} about computing them with \CMSO[2]-transductions taking as inputs
  $\{\edge\}$-relational structures. 
  
\subsection{Transducing weakly-partitive trees from weakly-partitive set systems}
\label{ssec:transduce partitive tree}
A set system~$(U,\mathcal{S})$ is \defd{weakly-partitive} if
for every two overlapping sets~$X,Y\in\mathcal{S}$,
the sets~$X\cup Y$, $X\cap Y$, $X\setminus Y$, and~$Y\setminus X$ belong to~$\mathcal{S}$.
It is \defd{partitive} if, moreover,
for every two overlapping sets~$X,Y\in\mathcal{S}$,
their symmetric difference, denoted \defd{$X\symdiff Y$},
also belongs to~$\mathcal{S}$.
By extension, a set family~$\mathcal{S}$ is called \defd{weakly-partitive} or \defd{partitive}
whenever $\left(\bigcup\mathcal{S},\mathcal{S}\right)$ is a set system\xspace%
which is weakly-partitive or partitive, respectively
(note that by the definition of set system we also require that~$\emptyset\notin\mathcal{S}$, $\bigcup\mathcal{S}\in\mathcal{S}$, and~$\{a\}\in\mathcal{S}$ for every~$a\in\bigcup\mathcal{S}$).
\begin{figure}
    \includegraphics[scale=0.65]{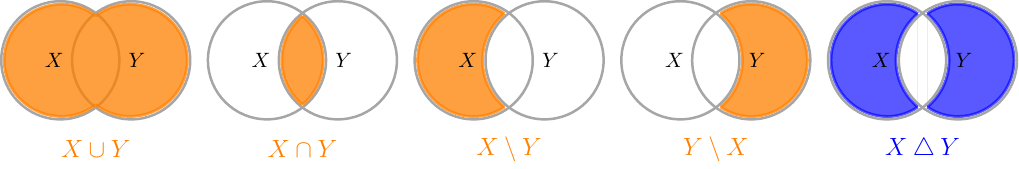}%
    \caption{%
        Sets to be included in a weakly-partitive set system (excludes the rightmost) or in a partitive set system (includes the rightmost) for two overlapping sets $X$ and $Y$.%
    }%
    \label{fig:overlapSet}
\end{figure}
\bigbreak

A member of a set system~$(U,\mathcal{S})$ is said to be~\defd{strong}
if it does not overlap any other set from~$\mathcal{S}$.
The sub-family~$\str(\mathcal{S})$ of strong sets of~$\mathcal{S}$ is thus laminar by definition.
Hence, it induces a laminar tree~$T$.
By extension,
we say that~$T$ is \defd{induced by the set system~$(U,\mathcal{S})$} (or simply by~$\mathcal{S}$).
The next result extends \autoref{thm:transduce laminar tree},
by showing that~$T$ can be \CMSO[2]-transduced from~$\mathcal{S}$.
\begin{lem}
    \label{lem:transduce laminar tree induced by weakly-partitive family}
    Let~$\Sigma$ be an extended vocabulary,
    including a unary set predicate name~$\SET$
    and not including the binary relational symbol~$\ancestor$.
    There exists a non-deterministic overlay \CMSO[2]-transduction~$\tau$
    such that,
    for each set system~$(U,\mathcal{S})$
    represented as the $\{\SET\}$-structure~$\mathbb{S}$
    and inducing the laminar tree~$T$ with ${L(T)=U}$,
    and for each $\Sigma$-structure~$\mathbb{A}$
    with~$\mathbb{S}\sqsubseteq\mathbb{A}$,
    $\tau(\mathbb{A})$ is non-empty
    and every output in~$\tau(\mathbb{A})$
    is equal to $\mathbb{A}\sqcup\mathbb{T}$ for some $\{\ancestor\}$-structure $\mathbb{T}$ representing $T$.
\end{lem}
\begin{proof}
    On the $\{\SET\}$-structure~$\mathbb{S}$,
    it is routine to define an \MSO-formula~$\phi_{\SSET}(Z)$
    that identifies those subsets~$Z\subseteq U$
    that are strong members of~$\mathcal{S}$:
    \begin{align*}
        \phi_{\SSET}(Z)&:=
        \SET(Z)
        \land
        \forall X\:
        \Big(
        \big( \SET(X) \land X\cap Z\neq\emptyset \big)
        \implies
        \big( X\subseteq Z \lor Z\subseteq X \big)
        \Big).
    \end{align*}
    Hence, we can design an \MSO-interpretation
    that outputs the $\Sigma\cup\{\SSET\}$-structure
    corresponding to~$\mathbb{A}$ equipped with the set unary predicate $\SSET$
    that selects strong members of~$\mathcal{S}$.
    Thus, up to renaming the predicates~$\SSET$ and~$\SET$,
    by \autoref{thm:transduce laminar tree},
    we can produce,
    through a \CMSO[2]-transduction,
    the $\Sigma\cup\{\ancestor\}$-structure $\mathbb{A}\sqcup\mathbb{T}$
    where~$\mathbb{T}$ is the tree-structure modelling the laminar tree~$T$ induced by~$\str(\mathcal{S})$
    with~$L(T)=U$
    (once the output obtained, the interpretation drops the set predicate~$\SSET$ which is no longer needed).
\end{proof}

Clearly, the laminar tree~$T$ of a weakly-partitive set system~$(U,\mathcal{S})$ does not characterize $(U,\mathcal{S})$.
However,
as shown by the \hyperref[thm:weakly-partitive tree]{below theorem},
a labelling of its inner nodes and a controlled partial ordering  of its nodes (which we define below)
are sufficient to characterize all the sets of~$\mathcal{S}$.
For~$Z$ a set equipped with a partial order~$<$
and~$X$ a subset of~$Z$,
we say that~$X$ is a \defd{$<$-interval}
whenever $<$ defines a total order on~$X$
and for every~$a,b\in X$ and every~$c\in Z$,
$a<c<b$ implies~$c\in X$.%
\begin{thmC}[\cite{chein1981partitive,raoThesis}]
    \label{thm:weakly-partitive tree}%
    Let~$\mathcal{S}$ be a weakly-partitive family,
    $\str(\mathcal{S})$ be its subfamily of strong sets,
    and~$T$ be the laminar tree it induces.
    There exists a total labelling function~$\lambda$
    from the set~$V(T)\setminus L(T)$ of inner nodes of~$T$
    to the set~$\{\DEGENERATE,\PRIME,\LINEAR\}$,
    and, for each inner node~$t\in\lambda^{-1}(\LINEAR)$,
    a linear ordering~$<_t$ of its children,
    such that
    every inner node having exactly two children is labelled by~$\DEGENERATE$ and the following two conditions are satisfied:
    \begin{itemize}[nosep]
        \item for each~$X\in\mathcal{S}\setminus\str(\mathcal{S})$,
            there exists~$t\in V(T)$
            and a subset~$\mathcal{C}$ of children of~$t$
            such that $X=\bigcup_{c\in\mathcal{C}}L(T_c)$
            and
            either~$\lambda(t)=\LINEAR$
            and~$\mathcal{C}$ is a $<_t$-interval,
            or~$\lambda(t)=\DEGENERATE$;
        \item conversely,
            for each inner node~$t$
            and each non-empty subset~$\mathcal{C}$ of children of~$t$,
            if
            either~$\lambda(t)=\LINEAR$
            and~$\mathcal{C}$ is a $<_t$-interval,
            or~$\lambda(t)=\DEGENERATE$,
            then~$\bigcup_{c\in\mathcal{C}}L(T_c)\in\mathcal{S}$.
    \end{itemize}
\end{thmC}
Furthermore, $T$ and $\lambda$ are uniquely determined from~$\mathcal{S}$,
and, for each inner node~$t$ of~$T$ labelled by~$\LINEAR$, only two orders~$<_t$ are possible,
one being the inverse of the other
(indeed, inverting an order~$<$ does preserve the property of being a $<$-interval).
Hence, every weakly-partitive family~$\mathcal{S}$ is characterized
by a labelled and partially-ordered tree~$(T,\lambda,<)$
where~$T$ is the laminar tree induced by the subfamily~$\str(\mathcal{S})$ of strong sets of~$\mathcal{S}$,
$\lambda:V(T)\setminus L(T)\to\{\DEGENERATE,\PRIME,\LINEAR\}$ is the labelling function,
and~$<$ is the partial order $\bigcup_{t\in\lambda^{-1}(\LINEAR)}<_t$ over~$V(T)$.
As, up-to inverting some of the $<_t$ orders, $(T,\lambda,<)$ is unique,
we abusively call it \defd{\strong{the} weakly-partitive tree induced by~$\mathcal{S}$}.
Conversely, a weakly-partitive tree characterizes the unique weakly-partitive set system which induced it.

We naturally model a weakly-partitive tree~$(T,\lambda,<)$ of a weakly-partitive set system $(U,\mathcal{S})$
by the $\{\ancestor,\DEGENERATE,\betweeness\}$-structure $\mathbb{T}$
of universe $U_{\mathbb{T}}=V(T)$ such that
$\struct{V(T),\ancestor_{\mathbb{T}}}\sqsubseteq\mathbb{T}$ models~$T$ with ${L(T)=U}$,
$\DEGENERATE_{\mathbb{T}}$ is a unary relation
which selects the inner nodes of~$T$ of label~$\DEGENERATE$,
\ie, ${\DEGENERATE_{\mathbb{T}}=\lambda^{-1}(\DEGENERATE)}$,
and~$\betweeness_{\mathbb{T}}$ is a ternary relation selecting triples $(x,y,z)$ satisfying $x<y<z$ or $z<y<x$.
(Although it is possible, through a non-deterministic \MSO-transduction, to define~$<$ from~$\betweeness_{\mathbb{T}}$,
the use of $\betweeness_{\mathbb{T}}$ rather than~$<_{\mathbb{T}}$ ensures uniqueness of the output weakly-partitive tree.)
The inner nodes of~$T$ that are labelled by~$\LINEAR$ can be recovered through an \MSO-formula as those inner nodes whose children are related by~$\betweeness$.
The inner nodes labelled by~$\PRIME$ can be recovered through an \MSO-formula as those inner nodes that are labelled neither by~$\DEGENERATE$ nor by~$\LINEAR$.
Using \autoref{thm:weakly-partitive tree},
it is routine to design an \MSO-transduction
which takes as input a weakly-partitive tree
and outputs the weakly-partitive set system which induces it.
The inverse \CMSO[2]-transduction is the purpose of the next result. 
\begin{thm}
    \label{thm:transduce weakly-partitive tree}%
    There exists a non-deterministic \CMSO[2]-transduction~$\tau$
    such that,
    for every weakly-partitive set system~$(U,\mathcal{S})$
    represented as the $\{\SET\}$-structure $\mathbb{S}$
    and inducing the weakly-partitive tree $(T,\lambda,<)$
    represented as the $\{\ancestor,\DEGENERATE,\betweeness\}$-structure $\mathbb{T}$,
    we have $\mathbb{T}\in\tau(\mathbb{S})$
    and every output in $\tau(\mathbb{S})$ is a weakly-partitive tree of $(U,\mathcal{S})$.
\end{thm}
\begin{proof}
    By \autoref{lem:transduce laminar tree induced by weakly-partitive family},
    we have a \CMSO[2]-transduction
    which on inputs~$\mathbb{S}$ outputs the $\{\SET,\ancestor\}$-structure $\mathbb{S}\sqcup\mathbb{T'}$,
    where~$\mathbb{T'}$ is the~$\{\ancestor\}$-structure modelling~$T$, with $L(T)=U=U_{\mathbb{S}}$.
    It thus remains to define a second \CMSO[2]-transduction
    which outputs~$\mathbb{T}$ from~$\mathbb{S}\sqcup\mathbb{T'}$.
    Actually, an \MSO-interpretation is enough.
    Indeed, we can define
    an \MSO-formula~$\phi_{\DEGENERATE}$ which selects the inner nodes of~$T$ that are labelled~$\DEGENERATE$ by~$\lambda$,
    and an \MSO-formula~$\phi_{\betweeness}$ which selects triples~$(x,y,z)$ of elements such that~$x<y<z$ or~$z<y<x$.
    In order to define both formulas,
    we use a functional symbol~$\leafset(t)$
    which is interpreted as the set~$L(T_t)$,
    and is clearly \MSO-definable on~$\{\ancestor\}$-structures such as~$\mathbb{T}$.
    Indeed, an element~$z$ belongs to~$\leafset(t)$ \ifof
    it is a leaf and has~$t$ as ancestor.

    According to \autoref{thm:weakly-partitive tree},
    an inner node~$t$ is labelled~$\DEGENERATE$
    \ifof,
    for every two children $s_1,s_2$ of $t$ the set $L(T_{s_1})\cup L(T_{s_2})$ is in $\mathcal{S}$. We can define this in \MSO using the following formula:
        $$\phi_{\DEGENERATE}(x):=\forall y \forall z\Big(\big(\parent(x,y)\land \parent(x,z)\big)\implies \SET\big(\leafset(y)\cup \leafset(z)\big)\Big).$$

    Now, according to \autoref{thm:weakly-partitive tree},
    we have~$x<y<z$ or~$z<y<x$ exactly when~$x$, $y$, and~$z$ are three distinct children of some node~$t$
    and the three following properties are satisfied:
    \begin{enumerate}[nosep]
        \item\label{it:not degenerate} $L(T_x)\cup L(T_z)$ is not a member of~$\mathcal{S}$;
        \item\label{it:not prime}\label{it:x+y-z} there exists a member of~$\mathcal{S}$ which includes both~$L(T_x)$ and~$L(T_y)$ but excludes~$L(T_z)$;
        \item\label{it:betweeness}\label{it:y+z-x} there exists a member of~$\mathcal{S}$ which includes both~$L(T_y)$ and~$L(T_z)$ but excludes~$L(T_x)$.
    \end{enumerate}
    Indeed,
    when these conditions are satisfied, $t$ is necessarily labelled $\LINEAR$
    (\autoref{it:not degenerate} asserts that it is not labelled by $\DEGENERATE$ and, \eg, \autoref{it:not prime} that it is not labelled by $\PRIME$),
    and, among the possible orderings of~$\{x,y,z\}$, only the ones ensuring~$x<y<z$ or~$z<y<x$ are possible
    ($x<z<y$ and~$y<z<x$ are not possible because of \autoref{it:x+y-z},
    and~$y<x<z$ and~$z<x<y$ are not possible because of \autoref{it:y+z-x}).
    Each of these properties is \MSO-definable:
    \begin{align*}
        \phi_{\betweeness}(x,y,z):=&
            \exists t\Big[ \parent(t,x)\land\parent(t,y)\land\parent(t,z)\land 
            \\*
            &\neg\SET\big(\leafset(x)\cup\leafset(z)\big)\land
            \\*
            &\exists S\Big(\SET(S)\land\big(\leafset(x)\cup\leafset(y)\big)\subseteq S\land\leafset(z)\cap S=\emptyset\Big)\land
            \\*
            &\exists S\Big(\SET(S)\land\big(\leafset(z)\cup\leafset(y)\big)\subseteq S\land\leafset(x)\cap S=\emptyset\Big)\Big].
    \end{align*}
    Once~$\phi_{\DEGENERATE}$ and~$\phi_{\betweeness}$ are defined,
    our transduction drops the original predicate~$\SET$
    which is no longer needed in the output~$\mathbb{T}$.
\end{proof}

If~$\mathcal{S}$ is partitive
then the weakly-partitive tree it induces enjoys a simple form,
and is unique.
Indeed, the label $\LINEAR$ and, thus, the partial order~$<$, are not needed.%
\begin{thmC}[\cite{chein1981partitive}]
    \label{thm:partitive tree}
    Let~$\mathcal{S}$ be a weakly-partitive family
    and~$(T,\lambda,<)$ be the weakly-partitive tree it induces.
    If~$\mathcal{S}$ is partitive,
    then~$\lambda^{-1}(\LINEAR)=\emptyset$ and~$<$ is empty.
\end{thmC}
Hence, in case of a partitive set systems~$(U,\mathcal{S})$,
we can consider the simpler object~$(T,\lambda)$,
called \defd{the partitive tree induced by~$\mathcal{S}$}
(or \defd{the partitive tree of~$(U,\mathcal{S})$})
in which~$\lambda$ maps $V(T)\setminus L(T)$ to $\{\DEGENERATE,\PRIME\}$.
As a direct consequence of \autoref{thm:partitive tree} and of \autoref{thm:transduce weakly-partitive tree},
we can produce, through a \CMSO[2]-transduction,
the partitive tree induced by a partitive set system
and naturally modelled by an $\{\ancestor,\DEGENERATE\}$-structure.
\begin{cor}
    \label{cor:transduce partitive tree}%
    There exists a non-deterministic \CMSO[2]-transduction~$\tau$
    such that,
    for each partitive set system~$(U,\mathcal{S})$
    represented as the $\{\SET\}$-structure~$\mathbb{S}$
    and inducing the partitive tree $(T,\lambda)$
    represented as the $\{\ancestor,\DEGENERATE\}$-structure $\mathbb{T}$,
    we have~$\mathbb{T}\in\tau(\mathbb{S})$
    and every output in $\tau(\mathbb{S})$ is a partitive tree of $(U,\mathcal{S})$.
\end{cor}

\subsection{Application to modular decomposition}\label{sec:modular-dec}
A construction similar to the one we provide here is also given in \cite{Courcelle06}, but based on graph operations. However, the two encodings
are equivalent \emph{w.r.t.} \MSO-transductions. Let~$G$ be a directed graph
and let~$M\subseteq V(G)$.
We say that~$M$ is a \defd{module (of~$G$)}
if for every $u\notin M$ and every $v,w\in M$,
${uv\in E(G)\!\iff\! uw\in E(G)}$
and~${vu\in E(G)\!\iff\! wu\in E(G)}$.
Clearly,
the empty set, $V(G)$, and all the singletons~$\{x\}$ for~$x\in V(G)$
are modules;
they are called \defd{the trivial modules} of~$G$.
We say a non-empty module~$M$ is \defd{maximal}
if it is not properly contained in any non-trivial module.
Furthermore, we use the notion of \emph{strong} modules to coincide with the strong sets in the set system consisting of all modules of $G$.
Let~$M$ and~$M'$ be two disjoint non-empty modules of~$G$.
Considering the edges that go from~$M$ to~$M'$,
namely edges from the set $(M\times M')\cap E(G)$,
we have two possibilities:
either it is empty, or it is equal to~$M\times M'$.
We write~\defd{$\nfullsucc{M}{M'}$} in the former case
and~\defd{$\fullsucc{M}{M'}$} in the latter.
(It is of course possible to have both~$\fullsucc{M}{M'}$ and~$\fullsucc{M'}{M}$.)
A~\defd{modular partition} of~$G$
is a partition $\mathcal{P}=\{M_1,\ldots,M_\ell\}$ of~$V(G)$
such that every~$M_i$ is a non-empty module.
A modular partition $\mathcal{P}=\{M_1,\dots,M_\ell\}$ is called \defd{maximal}
if it is non-trivial and every~$M_i$ is strong and maximal.
Note that every graph has exactly one maximal modular partition.

A \defd{modular decomposition} of~$G$
is a rooted tree~$T$
in which the leaves are the vertices of~$G$,
and for each inner node~$t\in T$,
$t$ has at least two children
and the set~$L(T_t)$ is a module of~$G$.
In a modular decomposition~$T$ of~$G$,
for each inner node~$t\in V(T)$
with children~$c_1,\ldots,c_r$,
the family $\mathcal{P}_t=\{L(T_{c_1}),\ldots,L(T_{c_r})\}$
is a modular partition of~$G[L(T_t)]$.
When each such partition is maximal,
the decomposition is unique
and it is called \defd{the maximal modular decomposition} of~$G$.
The maximal modular decomposition~$T$ of~$G$ alone
is not sufficient to characterize~$G$.
However, enriching~$T$ with,
for each inner node~$t$ with children~$c_1,\ldots,c_j$,
the information of which pair of modules $\big(L(T_{c_i}),L(T_{c_j})\big)$
is such that $\fullsucc{L(T_{c_i})}{L(T_{c_j}})$,
yields a unique canonical representation of~$G$.
Formally, the \defd{enriched modular decomposition} of~$G$
is the pair~$(T,F)$
where~$T$ is the maximal modular decomposition of~$G$ (with~$L(T)=V(G)$)
and~$F\subset V(T)\times V(T)$
is a binary relation,
that relates a pair~$(s,t)$ of nodes of~$T$,
denoted $st\in F$,
exactly when~$s$ and~$t$ are siblings and~$\fullsucc{L(T_s)}{L(T_t)}$.
The elements of~$F$ are called \defd{$\mathsf{m}$-edges}.

It should be mentioned that the family of all non-empty modules of~$G$
is known to be weakly-partitive (or even partitive when~$G$ is undirected).
In particular, the maximal modular decomposition~$T$ of~$G$
is the laminar tree induced by the family of strong modules.
Hence,  \autoref{thm:transduce weakly-partitive tree}
could be used to produce a partially-ordered and labelled tree which displays all the modules of~$G$.
However, this weakly-partitive tree is not sufficient
for being able to recover the graph~$G$ from it.
We now prove how to obtain the enriched modular decomposition of~$G$
through a \CMSO[2]-transduction.

To model enriched modular decompositions as relational structures
we use the relational vocabulary $\{\ancestor, \medge\}$
where $\ancestor$ and $\medge$ are two binary relation names.
An enriched modular decomposition $(T,F)$ of a graph $G$
is modelled by the {$\{\ancestor, \medge\}$}-structure $\mathbb{M}$
with universe $U_{\mathbb{M}}=V(T)$,
$\ancestor_{\mathbb{M}}$ being the set of pairs $(s,t)$
for which~$s$ is an ancestor of~$t$ in~$T$,
and $\medge_{\mathbb{M}}$ being the set of all pairs $(s,t)$
such that $st\in F$
(in particular, $s$ and~$t$ are siblings in~$T$).
It is reasonably straightforward to transduce the maximal modular decomposition from a graph using \autoref{lem:transduce laminar tree induced by weakly-partitive family}. For the sake of completeness we give a proof below.
\begin{cor}
    \label{thm:transduce modular decomposition}
    There exists a non-deterministic \CMSO[2]-transduction $\tau$
    such that for every directed graph $G$ represented as the $\{\edge\}$-structure $\mathbb{G}$,
    $\tau(\mathbb{G})$ is non-empty
    and every output in~$\tau(\mathbb{G})$ is equal to some $\{\ancestor,\medge\}$-structure $\mathbb{M}$
    representing the enriched modular decomposition~$(T,F)$ of $G$.
\end{cor}
\begin{proof}
    Let $G$ be a graph represented by the $\{\edge\}$-structure $\mathbb{G}$.
    Several objects are associated to~$G$,
    and each of them can be described by a structure:
    \begin{itemize}[nosep]
        \item let~$\mathcal{M}$ be the family of non-empty modules of~$G$
            and let~$\mathbb{S}$ be the $\{\SET\}$-structure
            modelling the weakly-partitive set system $(V(G),\mathcal{M})$
            with~$U_{\mathbb{S}}=V(G)$;
        \item let~$T$ be the laminar tree induced by the weakly-partitive family~$\mathcal{M}$
            and let~$\mathbb{T}$ be the $\{\ancestor\}$-structure
            modelling it with~$U_{\mathbb{T}}=V(T)$ and~$L(T)=V(G)\subset U_{\mathbb{T}}$;
        \item let~$F$ be the $\mathsf{m}$-edge relation,
            namely the subset of~$V(T)\times V(T)$ such that~$(T,F)$
            is the maximal modular decomposition of~$G$,
            and let~$\mathbb{M}$ be the $\{\ancestor,\medge\}$-structure
            modelling it with~$\mathbb{T}\sqsubset\mathbb{M}$ and $U_{\mathbb{T}}=U_{\mathbb{M}}$.
    \end{itemize}
    Our \CMSO[2]-transduction is obtained by composing the three following transductions:
    \begin{itemize}[nosep]
        \item$\tau_1$: an \MSO-interpretation which outputs the $\{\edge,\SET\}$-structure $\mathbb{G}\sqcup\mathbb{S}$ from~$\mathbb{G}$;
        \item$\tau_2$: the non-deterministic \CMSO[2]-transduction given by \autoref{lem:transduce laminar tree induced by weakly-partitive family}
            which produces the $\{\edge,\SET,\ancestor\}$-structure $\mathbb{G}\sqcup\mathbb{S}\sqcup\mathbb{T}$ from $\mathbb{G}\sqcup\mathbb{S}$;
        \item$\tau_3$: an \MSO-interpretation which outputs the $\{\ancestor,\medge\}$-structure $\mathbb{M}$ from $\mathbb{G}\sqcup\mathbb{S}\sqcup\mathbb{T}$.
    \end{itemize}

    In order to define $\tau_1$ it is sufficient to observe that
    there exists an \MSO-formula $\phi_{\SET}(Z)$
    with one monadic free-variable,
    which is satisfied exactly when~$Z$ is a non-empty module of~$G$.
    Then, since~$\tau_2$ is given by \autoref{lem:transduce laminar tree induced by weakly-partitive family},
    it only remains to define~$\tau_3$.
    Given an inner node~$t$,
    we can select, within \MSO,
    the set~$L(T_t)$ of leaves of the subtree rooted in~$t$.
    We thus assume a function~$\leafset$,
    with one first-order free-variable
    which returns the set of leaves of the subtree rooted at the given node.
    Equipped with this function,
    we can define the \MSO-formula $\phi_{\medge}$
    which selects pairs~$(s,r)$ of siblings
    such that~$sr\in F$.
    Remember that this happen exactly when there exist $u\in L(T_s)$ and~$v\in L(T_r)$
    such that~$uv\in E(G)$.
    Hence, $\phi_{\medge}$ could be defined as:
    \begin{align*}
        \phi_{\medge}(s, r):=
            s\neq r
            \:\land\:
            &
            \exists t\:
            \big(\parent(t,s)\land\parent(t,r)\big)\land
            \\*&
            \exists x\exists y\big(x\in\leafset(s)\land y\in\leafset(r)\land \edge(x,y)\big).
    \end{align*}
    Once defined, the \MSO-interpretation~$\tau_3$
    simply drops all non-necessary relations and predicates
    (namely~$\edge$ and~$\SET$)
    and keeps only the~$\ancestor$ and~$\medge$ relations.
\end{proof}

Notice that it is routine to design a deterministic \MSO-transduction which,
given an $\{\ancestor,\medge\}$-structure $\mathbb{M}$
representing an enriched modular decomposition of some directed graph~$G$,
produces the $\{\edge\}$-structure $\mathbb{G}$ representing~$G$.\\

In the following we provide an example of a natural property that can be \CMSO[2]-defined easily using the modular decomposition and hence (though backwards translation) is \CMSO[2]-definable on graphs.
\begin{exa}\label{ex:numModules}
We want to define the property of an undirected graph $G$ having an even number of modules in \CMSO[2]. Let $T$ be the modular decomposition of $G$ and $\mathcal{M}$ the set of modules of $G$. As $\mathcal{M}$ is partitive, we can understand $T$ as a partitive tree  and let $\lambda: V(T)\rightarrow \{\PRIME, \DEGENERATE\}$ be the labelling of $T$ that exists due to \autoref{thm:weakly-partitive tree} and \autoref{thm:partitive tree}. By \autoref{thm:weakly-partitive tree} it holds that the number of modules of $G[L(T_t)]$ for a node $t$ with children $s_1,\dots, s_\ell$ for which $\lambda(t)=\PRIME$ is exactly the sum of the numbers of modules minus $1$ (to avoid double counting the empty set) of the $G[L(T_{s_i})]$ plus $2$ (which counts $L(T_t)$ and $\emptyset$). On the other hand, if $\lambda(t)=\DEGENERATE$ then the number of modules of $G[L(T_t)]$ is exactly the product of the numbers of modules of the $G[L(T_{s_i})]$ (here the modules $L(T_t)$ and $\emptyset$ are automatically counted precisely once). 

Hence, given a set $X$ of nodes of $T$ we can check whether $X$ is precisely the set of nodes $t$ for which $G[L(T_t)]$ has an even number of modules in \CMSO[2]. For this we first check for each node $t$ their labelling. If $\lambda(t)=\PRIME$ we check that $t\in X$ \ifof the parity of the number of children of $t$ that are in $X$ is even (counting exactly the children for which the number of modules minus $1$ of $G[L(T_{s_i})]$ is odd). 
If $\lambda(t)=\DEGENERATE$ we check that $t\in X$ \ifof  there is a child $s$ of $t$ which is in $X$ (in this case the product of the number of modules of $G[L(T_{s_i})]$ is even).   Note that additionally leaves of $T$ are included in $X$ as a graph with one vertex contains exactly two modules. The number of modules of $G$ is even \ifof there is a set $X$ which consist of all nodes $t$ for which $G[L(T_t)]$ is even and the root of $T$ is in $X$.

To define the sentence, we let $\mathbb{T}$ be the $\{\ancestor,\PRIME,\DEGENERATE\}$-structure modelling $T$ equipped with $\lambda$. While we did not provide a transduction producing this particular structure, it is easy to modify the above transductions to obtain the desired structure.  We further use auxiliary predicates $\children(Y,y)$ expressing that $Y$ is the set of children of $y$, $\rootT(y)$ expressing that $y$ is the root of $T$ and $\leaf(y)$ expressing that $y$ is a leaf of $T$, which are routine to implement.
Hence, we can define a sentence $\phi$ that is satisfied by any graph that has an even number of modules as follows.
   \begin{align*}
        \phi:=&\exists X \forall y \exists Y \Bigg(\children(Y,y)\land \Big[\rootT(y)\rightarrow y\in X\Big]\land \Big[\leaf(y)\rightarrow y\in X\Big]\land \\
        &\Big[\Big(\big(\PRIME(y)\land \C[2](X\cap Y)\big)\lor \big(\DEGENERATE(y)\land  X\cap Y\neq \emptyset\big) \Big)\leftrightarrow y\in X \Big]\Bigg).
    \end{align*}
\end{exa}

\paragraph*{Cographs}
Let~$G$ be a graph, $(T,F)$ be its modular decomposition,
and~$(T,\lambda,<)$ be the weakly-partitive tree induced by the (weakly-partitive) family of its modules.
Let~$t$ be an inner node of~$T$, let~$\mathcal{C}$ be its set of children,
and let~$C$ be the graph~$\big(V(T)\setminus L(T), F\big)\left[\mathcal{C}\right]$ induced by~$F$ on the set of children of~$t$.
It can be checked that,
if~$\lambda(t)=\DEGENERATE$, then~$C$ is either a clique or an independent set,
and if~$\lambda(t)=\LINEAR$, then~$C$ is a \emph{tournament consistent with~$<_t$}
(\ie, for every $x,y\in\mathcal{C}$, $xy$ is an edge of~$C$ \ifof $x<_t y$)
or with the inverse of~$<_t$.
In the former case,
we can refine the $\DEGENERATE$ label into~$\SERIES$ and~$\PARALLEL$ labels,
thus expressing that~$C$ is a clique or an independent set, respectively.
In the latter case, up-to reversing~$<_t$,
we can ensure that~$C$ is a tournament consistent with~$<_t$.
This yields a refined weakly-partitive tree~$(T,\gamma,<)$,
where~$\gamma$ maps inner nodes to $\{\SERIES,\PARALLEL,\PRIME,\LINEAR\}$
and~$<$ is the order $\bigcup_{t\in\gamma^{-1}(\LINEAR)}<_t$
which ensures that tournaments are consistent with the corresponding $<_t$.
Notice that this labelled and partially-ordered tree is now uniquely determined from~$G$.
Moreover, edges from~$F$ that connect children of a node not labelled by~$\PRIME$
can be recovered from the so-refined weakly-partitive tree.
In particular, if no nodes of~$T$ is labelled by~$\PRIME$,
$(T,F)$ and thus~$G$ is fully characterized by~$(T,\gamma,<)$.
Graphs for which this property holds
are known as \defd{directed cographs}. Directed cographs equivalently can be defined by a finite set of forbidden subgraphs (see \cite{CrespelleP06}) and in the undirected case, cographs are exactly the $P_4$-free graphs. Directed cographs
can be described by the refined weakly-partitive tree, called \defd{cotree}, explained above and formalized in the following statement.
\begin{thm}
    \label{thm:cotree}
    Let~$G$ be a directed cograph and let~$T$ be the laminar tree induced by the family of its strong modules.
    There exists a unique total labelling~$\lambda$ from the set~$V(T)\setminus L(T)$ of inner nodes of~$T$
    to the set $\{\SERIES,\PARALLEL,\LINEAR\}$ of labels,
    and, for each inner node $t\in\lambda^{-1}(\LINEAR)$,
    a unique linear ordering~$<_t$ of its children,
    such that every inner node having exactly two children is labelled $\SERIES$ or $\PARALLEL$,
    and the following condition is satisfied:
    \begin{itemize}[nosep]
        \item for every two leaves~$x$ and~$y$ of~$T$,
            denoting by~$t$ their least common ancestor and $s_x,s_y$ the children of $t$ that are ancestors of $x,y$ respectively,
            $xy$ is an edge of~$G$
            \ifof
            either~$t$ is labelled by~$\LINEAR$ and~$s_x<_ts_y$,
            or~$t$ is labelled by~$\SERIES$.
    \end{itemize}
\end{thm}

We naturally model cotrees as $\{\ancestor,\SERIES,\epsord\}$-structures as follow.
A cotree $(T,\gamma,<)$ is modelled by~$\mathbb{C}$
where~$U_{\mathbb{C}}=V(T)$, $\SERIES_{\mathbb{C}}=\gamma^{-1}(\SERIES)$, and $\epsord_{\mathbb{C}}=\{(x,y)\mid x<y\}$.
The nodes that are labelled by $\LINEAR$ could be recovered as those inner nodes whose children are related by $\epsord$,
while the nodes that are labelled by $\PARALLEL$ could be recovered as those inner nodes which are labelled neither by $\SERIES$ nor by $\LINEAR$.
Based on \autoref{thm:cotree} and as a consequence of \autoref{thm:transduce modular decomposition},
we can design a \CMSO[2]-transduction which produces the cotree of a cograph~$G$ from~$G$.
\begin{cor}
    \label{thm:transduce cotree}
    There exists a non-deterministic \CMSO[2]-transduction~$\tau$
    such that, for each directed cograph~$G$ modelled by the $\{\edge\}$-structure~$\mathbb{G}$,
    $\tau(\mathbb{G})$ is non-empty
    and every output in~$\tau(\mathbb{G})$ is equal to some $\{\ancestor,\SERIES,\epsord\}$-structure~$\mathbb{C}$
    representing the cotree~$(T,\gamma,<)$ of~$G$.
\end{cor}
As for undirected graphs where the set system of modules is partitive, a simpler result in which the label $\LINEAR$ is not needed holds for undirected cographs. A
corollary of \autoref{thm:transduce cotree} is that on subclasses of cographs recognisability and \CMSO[2]-definability are equivalent, which was not
known. We will explain this corollary in a more general setting using split-decomposition in \autoref{sec:restricstedRankWidth}.

\section{Transducing split and bi-join decompositions from graphs}\label{sec:split}
In this section we consider systems of bipartitions. In graph theory, two known systems of bipartitions are splits and bi-joins. Both are instances of so called
weakly-bipartitive systems of bipartitions (or partitive systems of bipartitions in case of undirected graphs). In this section, we first give a
\CMSO[2]-transduction producing the canonical tree-like decomposition of weakly-bipartitive systems of bipartitions and then derive
\CMSO[2]-transductions for producing split and bi-join decompositions of graphs. We also emphasize that Courcelle provides a similar transduction for the split
  decomposition of strongly connected directed graphs, but using order-invariant \MSO-formulas \cite{Courcelle06}. Again, a corollary of our construction is that we can
  compute split decompositions of strongly connected directed graphs using only \CMSO[2]-formulas and on the basis of the edge
  relation. This answers another question from \cite{CourcelleT12}.
\subsection{Transducing weakly-bipartitive trees from weakly-bipartitive systems of bipartitions}
A \emph{bipartition system} is a pair $(U,\mathcal{B})$ consisting of a finite set $U$, the \emph{universe}, and a family $\mathcal{B}$ of bipartitions of $U$ such that $\{\emptyset,U\}\notin\mathcal{B}$ and
$\{\{a\},U\setminus \{a\}\}\in \mathcal{B}$ for all $a\in U$.

To model bipartition systems, we use the extended vocabulary~$\{\BIPARTITION\}$
where $\BIPARTITION$ is a unary set predicate name.
A bipartition system~$(U,\mathcal{B})$ is thus naturally modeled
as the $\{\BIPARTITION\}$-structure $\mathbb{B}$
with universe~$U_{\mathbb{B}}=U$
and interpretation~$\{\{X,U\setminus X\}\mid X\in \BIPARTITION_{\mathbb{B}}\}=\mathcal{B}$.

Two bipartitions $\{X,Y\}$ and $\{X',Y'\}$ of a set $U$ \emph{overlap} if the four sets $X\cap X'$, $X\cap Y'$, $Y\cap X'$ and $Y\cap Y'$ are non-empty.
A bipartition family is said to be \emph{laminar} (aka \emph{overlap-free})
if no two bipartitions in $\mathcal{B}$ overlap.
By extension, we call a family $\mathcal{B}$ of bipartitions of a set $U$ \emph{laminar}
whenever $(U,\mathcal{B})$ is a bipartition system which is laminar.
Similarly to laminar set systems, we can associate a tree with each laminar bipartition system. The tree $T$, called the \emph{laminar tree induced by $(U,\mathcal{B})$} (or \emph{laminar tree of $(U,\mathcal{B})$}), associated with a laminar bipartition system $(U,\mathcal{B})$ is unrooted, each element of $U$ corresponds to a leaf of $T$ and for each bipartition $\{X,Y\}$ in $\mathcal{B}$ there is exactly one edge $e$ such that $X=L(T_1)$ and $Y=L(T_2)$ for the two connected components $T_1,T_2$ of $T-e$. We remark that there is a unique such tree and each inner node has degree at least three.
Furthermore, the size of the laminar tree is linearly bounded in the size of the universe $U$.

A bipartition system $(U,\mathcal{B})$ is said to be \emph{weakly-bipartitive} if for every two overlapping bipartitions $\{X,Y\},\{X',Y'\}\in \mathcal{B}$, the biparitions $\{X\cup X',Y\cap Y'\}$, $\{X\cup Y',Y\cap X'\}$, $\{Y\cup X', X\cap Y'\}$ and $\{Y\cup Y', X\cap X'\}$ are in $\mathcal{B}$.
It is \emph{bipartitive}, if additionally for every two overlapping bipartitions $\{X,Y\},\{X',Y'\}\in \mathcal{B}$,
the bipartition $\{X\symdiff X', Y\symdiff X'\}$ is also in $\mathcal{B}$.
By extension, we call a family $\mathcal{B}$ of bipartitions of a set $U$ \emph{weakly-bipartitive} or \emph{bipartitive} whenever $(U,\mathcal{B})$ is a bipartition system which is weakly-bipartitive or partitive, respectively.
\begin{figure}
    \includegraphics[scale=0.65]{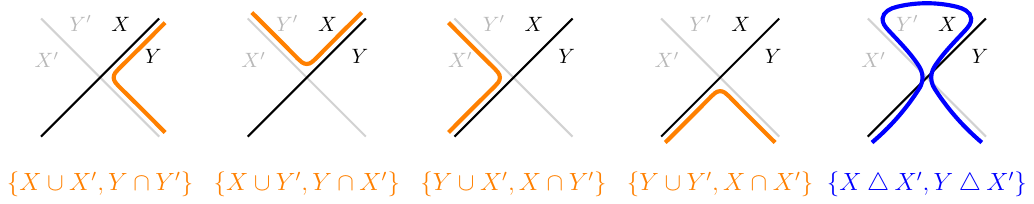}%
    \caption{%
        Bipartitions to be included in a weakly-bipartitive system of bipartitions (excludes the rightmost) or in a bipartitive  system of bipartitions (includes the rightmost) for two overlapping bipartitions $\{X,Y\}$ and $\{X',Y'\}$.%
    }%
    \label{fig:overlapBipartition}
\end{figure}
\bigbreak

A bipartition $\{X,Y\}$ of a bipartition system~$(U,\mathcal{B})$ is said to be \emph{strong} if $\{X,Y\}$ does not overlap with any other bipartition in $\mathcal{B}$. We denote the subfamily of strong bipartitions in $\mathcal{B}$ by $\str(\mathcal{B})$ and remark that $\str(\mathcal{B})$ is a laminar family.
Hence, it induces a laminar tree~$T$.
By extension,
we say that~$T$ is \defd{induced by the bipartition system~$(U,\mathcal{B})$} (or simply by~$\mathcal{B}$).  Note that $\mathcal{B}$ is not required to be laminar.
As the laminar tree $T$ of a bipartition system~$(U,\mathcal{B})$ is undirected, we model $T$ by the $\{\tedge\}$-structure $\mathbb{T}$ with universe $U_\mathbb{T}=V(T)$ and binary relation $\tedge_\mathbb{T}$ modelling the edge relation of $T$.

The following lemma gives a connection between laminar trees of set systems and laminar trees of bipartition systems. Note that a very similar statement follows from \cite[Lemma 1.14]{raoThesis} and hence we only provide a proof for sake of completeness.
\begin{lem}\label{prop:bipartitiveVSpartitive}
    Let $(U,\mathcal{B})$ be a bipartition system and $a\in U$.
    Then $(U, \mathcal{S}_a)$ is a set system where $\mathcal{S}_a:=\{X: a\notin X \text{ and } \{X,U\setminus X\} \in \mathcal{B}\}\cup \{\{a\},U\setminus\{a\}\}$.
    Moreover, the laminar tree of $(U,\mathcal{B})$ can be obtained from the laminar tree of $(U,\mathcal{S}_a)$ by deleting the root, making node $\{a\}$ adjacent to the node $U\setminus\{a\}$ and making the tree undirected.
\end{lem}
\begin{proof}
    Let $\mathcal{B}$ be a family of bipartitions of a set $U$ and $a\in U$ any element.
    First observe that $(U,\mathcal{S}_a)$ is a set system
    as $\emptyset\notin \mathcal{S}_a$ (since $\{\emptyset,U\}\notin \mathcal{B}$),
    $U\in \mathcal{S}_a$,
    and $\{b\}\in \mathcal{S}_a$ for every $b\in U$ (since~$\{a\}\in\mathcal{S}_a$ and for~$b\neq a$, $\{\{b\},U\setminus \{b\}\}\in \mathcal{B}$).

    To argue the second assertion, assume that $T$ is the laminar tree of the set system $(U,\mathcal{S}_a)$. First observe that in the set system $(U,\mathcal{S}_a)$ the element $a$ is not contained in any set apart from $U$ and $\{a\}$ and hence the leaf $\{a\}$ is a child of the root $U$ of $T$. Additionally, the set $U\setminus \{a\}$ is in $S_a$ as the bipartition $(\{a\},U\setminus \{a\})$ is in $\mathcal{B}$ by definition and hence $U\setminus \{a\}$ is also a child of the root $U$ of $T$. Let $T'$ be the tree obtained from $T$ by removing node $U$, making all edges undirected and adding an edge from node $\{a\}$ to node $U\setminus \{a\}$. Indeed, by our previous observations, $T'$ is an undirected tree. Furthermore, there is a correspondence between elements of $U$ and the leaves of $T'$ (element $b\in U$ corresponds to leaf $\{b\}$ of $T'$). To check that $T'$ is indeed the laminar tree of $T$, we need to verify that for each bipartition $\{X,Y\}$ in $\mathcal{B}$ there is exactly one edge $e$ of $T'$ such that $X=L(T'_1)$ and $Y=L(T'_2)$ for the two connected components $T'_1,T'_2$ of $T'-e$. We say that the edge \emph{$e$ implements the bipartition $\{X,Y\}$}.
    Assume that $\{X,Y\}$ is any bipartition from $\mathcal{B}$. First consider the case that $\{X,Y\}$ is the bipartition $\{\{a\},U\setminus\{a\}\}$. Then $\{X,Y\}$ is clearly implemented by the edge $\{\{a\},U\setminus\{a\}\}$ of $T'$. Now assume that $\{X,Y\}$ is any other bipartition in $\mathcal{B}$ and assume, without loss of generality, that $a\in Y$. Hence, $X\in \mathcal{S}_a$. Let $X'$ be the parent of $X$ in $T$. Note that $X$ must be a proper subset of $U\setminus \{a\}$ as by our assumption $\{X,Y\}$ is not the partition $\{\{a\},U\setminus \{a\}$ and hence $X'$ cannot be $U$. As we only deleted edge incident to $U$ in our construction of $T'$, $\{X,X'\}$ is an edge in $T'$ which clearly implements the bipartition $\{X,Y\}$. Finally, every bipartition $\{X,Y\}\in \mathcal{B}$ is clearly implemented by only one edge as every inner node of the tree $T'$ has degree at least three. Therefore, $T'$ is the laminar tree of $(U,\mathcal{B})$.
\end{proof}

\begin{lem}
    \label{lem:transduce laminar tree induced by weakly-bipartitive family}
    Let~$\Sigma$ be an extended vocabulary,
    including a unary set predicate name~$\BIPARTITION$
    and not including the binary relational symbol~$\tedge$.
    There exists a non-deterministic overlay \CMSO[2]-transduction~$\tau$
    such that,
    for each bipartition system~$(U,\mathcal{B})$
    represented as the $\{\BIPARTITION\}$-structure~$\mathbb{B}$
    and inducing the laminar tree~$T$ with ${L(T)=U}$,
    and for each $\Sigma$-structure~$\mathbb{A}$ with~$\mathbb{B}\sqsubseteq\mathbb{A}$,
    $\tau(\mathbb{A})$ is non-empty
    and every output in~$\tau(\mathbb{A})$ is equal to $\mathbb{A}\sqcup \mathbb{T}$ for some $\{\tedge\}$-structure $\mathbb{T}$ representing~$T$.
\end{lem}
\begin{proof}
    Let $\mathcal{B}$ be a family of bipartitions of $U$ represented as the extended relational structure~$\mathbb{B}$ and $\mathbb{A}$ a $\Sigma$-structure
    with $\mathbb{B}\sqsubseteq \mathbb{A}$. We first use the following two simple atomic \CMSO[2]-transductions to obtain a set system from $(U,\mathcal{B})$.
    Guess any colouring $A$ for which there is a single $a\in U$ with $A(a)$ being satisfied. Interpret using the following formula to obtain a $\SET$-predicate:
    $$\phi_\SET(X):=\forall a (A(a)\implies (X=\{a\} \lor X=U \lor a\notin X))\land (\BIPARTITION(X) \lor \BIPARTITION(U\setminus X)).$$
    The resulting set system $(U,\mathcal{S}_a)$ represented by the $\{\SET\}$-structure $\mathbb{S}$ is a set system by \autoref{prop:bipartitiveVSpartitive}. Hence, we can use the transduction from \autoref{lem:transduce laminar tree induced by weakly-partitive family} to obtain the $\Sigma\cup\{\ancestor\}$-structure $\mathbb{A}\sqcup \mathbb{T}$ where $\mathbb{T}$ is the $\{\ancestor\}$-structure representing the laminar tree $T$ of $(U,\mathcal{S}_a)$. We now use the following two atomic transductions to obtain the weakly-bipartitive tree of $(U,\mathcal{B})$ according to \autoref{prop:bipartitiveVSpartitive}. We first restrict the universe to all nodes which are not the root of $\mathbb{T}$. We then interpret using the following formula to obtain the $\tedge$-predicate which makes the graph undirected and adds an edge from $a$ to the node representing the set $U\setminus \{a\}$ which is the root of $T$ after removing the node corresponding to $U$:
    \begin{align*}
    \phi_\tedge(x,y):=&\parent(x,y)\lor \parent(y,x)\lor \\&(x=a\land y\not= a \land \lnot \exists z \parent(z,y))\lor (x\not= a\land y=a \land \lnot \exists z \parent(z,x)).
    \end{align*}
    Finally, the transduction keeps only $\tedge$ and any $Q\in \Sigma$
    as relations for its output. We do not need to filter the outputs as any guess for $A$ yields a valid laminar tree.  
\end{proof}

Similar to the set system case, for each weakly-bipartitive family $\mathcal{B}$ we can equip the laminar tree with a labelling of its inner nodes and a partial order which
characterizes the family $\mathcal{B}$.
We remind the reader that for a node $t$ of a tree $T$ with neighbour $s$ we denote by $T^t_s$ the connected component of $T-t$ containing $s$.

\begin{thmC}[{\cite[Theorem 3]{montgolfierThesis}}]\label{thm:weakly-bipartitive tree}
    Let~$\mathcal{B}$ be a weakly-bipartitive family,
    $\str(\mathcal{B})$ be its subfamily of strong bipartitions,
    and~$T$ the laminar tree it induces.
    There exists a total labelling function~$\lambda$ from the set $V(T)\setminus L(T)$
    of the inner nodes of~$T$ to the set $\{\DEGENERATE,\PRIME,\LINEAR\}$, and, for each inner node $t\in \lambda^{-1}(\LINEAR)$, a linear ordering $<_{t}$ of its neighbours, such that every inner node of degree exactly three is labeled by $\DEGENERATE$ (nodes of degree greater than three can be labelled $\DEGENERATE$, $\PRIME$ or $\LINEAR$) and the following conditions are satisfied:
    \begin{itemize}[nosep]
        \item for each bipartition~$\{X,Y\}\in\mathcal{B}\setminus\str(\mathcal{B})$,
            there exists~$t\in V(T)\setminus L(T)$
            and a subset~$\mathcal{C}$ of neighbours of~$t$
            such that either $X$ or $Y$ is equal to the set
            $\bigcup_{c\in\mathcal{C}}L(T^t_c)$ and either $\lambda(t)=\LINEAR$ and $\mathcal{C}$ is a $<_t$-interval, or $\lambda(t)=\DEGENERATE$;
        \item conversely,
            for each inner node~$t$
            and each non-empty subset~$\mathcal{C}$, $\mathcal{C}\not= U$ of neighbours of~$t$,
            the bipartition
            $\big\{\bigcup_{c\in\mathcal{C}}L(T^t_c),\bigcup_{c\notin\mathcal{C}}L(T^t_c)\big\}$ is a member of~$\mathcal{B}$ if either $\lambda(t)=\LINEAR$ and $\mathcal{C}$ is $<_t$-interval or $\lambda(t)=\DEGENERATE$.
    \end{itemize}
\end{thmC}
The tree $T$ and its labelling function $\lambda$ are uniquely determined by $\mathcal{B}$. We note that the total orders $<_t$ are uniquely determined by
$\mathcal{B}$ up to inverting and cyclic shifting, \ie, for a linear order $<$ of $X=\{x_1,\dots,x_\ell\}$ with $x_1<\dots<x_\ell$ we call the order $<'$ of $X$
with $x_\ell<'x_1<'\dots<' x_{\ell-1}$ a \defd{cyclic shift of $<$}. To see this, note that in the theorem above the bipartitions obtained from nodes with label
$\LINEAR$ can be equivalently obtained from a $<_t$-interval or the complement of a $<_t$-interval. Indeed, a cyclic shifting of a linear order $<$ maintains the property of being either a $<$-interval or the complement of a $<$-interval.
Neglecting that $<_t$ are only unique up to inverting and cyclic shifting, we represent every weakly-bipartitive family $\mathcal{B}$ by a triple $(T,\lambda,(<_t)_{t\in \lambda^{-1}(\LINEAR)})$ where $T$ is the laminar tree induced by $\str(\mathcal{B})$, $\lambda:V(T)\setminus L(T)\to\{\DEGENERATE,\PRIME,\LINEAR\}$ is the labelling function, and $<_t$ is the total order of the neighbours of $t$ described above. Note that in this case $\bigcup_{t\in\lambda^{-1}(\LINEAR)}<_t$ is not a partial order as it is not necessarily transitive.
We call the triple $(T,\lambda,(<_t)_{t\in \lambda^{-1}(\LINEAR)})$ from \autoref{thm:weakly-bipartitive tree} \defd{\strong{the} bipartitive tree} of $\mathcal{B}$.

We model weakly-bipartitive trees~$(T,\lambda,(<_t)_{t\in \lambda^{-1}(\LINEAR)})$ of a weakly-bipartitive family $\mathcal{B}$ of bipartition of $U$
by the $\{\tedge,\DEGENERATE,\cross\}$-structure $\mathbb{T}$
of universe $U_{\mathbb{T}}=V(T)$ such that
$\struct{V(T),\tedge_{\mathbb{T}}}\sqsubseteq\mathbb{T}$ models~$T$ with ${L(T)=U}$,
$\DEGENERATE_{\mathbb{T}}$ is a unary relation
which selects all inner nodes of~$T$ of label~$\DEGENERATE$,
\ie, ${\DEGENERATE_{\mathbb{T}}=\lambda^{-1}(\DEGENERATE)}$,
and~$\cross_{\mathbb{T}}$ is a relation of arity five selecting all 5-tuples $(t,x_1,x_2,x_3,x_4)$ for which the chord $x_1x_3$ crosses the chord $x_2x_4$ in $<_t$ understood as a cyclic order. Recall that nodes with less than four children cannot be labelled $\LINEAR$ and therefore the number of arguments of $\cross$ is not a restriction.
Similarly to the weakly-partitive case, we can recover the inner nodes of~$T$ that are labelled by~$\LINEAR$ through an \MSO-formula as those inner nodes whose neighbours are related by~$\cross$.
Then $\PRIME$ nodes are the ones which are labelled neither by~$\DEGENERATE$ nor by~$\LINEAR$ which can be obtained by an $\MSO$-transduction.
Furthermore, \autoref{thm:weakly-bipartitive tree},
gives rise to an \MSO-transduction
which takes as input a weakly-bipartitive tree
and outputs the weakly-bipartitive set system which induced it.

We can now extend the transduction from \autoref{lem:transduce laminar tree induced by weakly-bipartitive family} to obtain a \CMSO[2]-transduction producing the weakly-bipartitive tree of a weakly-bipartitive family.
\begin{thm}
    \label{thm:transduce weakly-bipartitive tree}
    There exists a non-deterministic \CMSO[2]-transduction~$\tau$
    such that,
    for every weakly-bipartitive bipartition system~$(U,\mathcal{B})$
    represented as the $\{\BIPARTITION\}$-structure~$\mathbb{B}$,
    $\tau(\mathbb{B})$ is non-empty and every output in~$\tau(\mathbb{B})$
    is equal to some $\{\tedge,\DEGENERATE,\cross\}$-structure~$\mathbb{T}$
    representing the weakly-bipartitive tree~$\left(T,\lambda,(<_t)_{t\in \lambda^{-1}(\LINEAR)}\right)$ of~$(U,\mathcal{B})$.
\end{thm}
\begin{proof}
    \autoref{lem:transduce laminar tree induced by weakly-bipartitive family} already provides a \CMSO[2]-transduction
    which on input $\mathbb{B}$ outputs the $\{\BIPARTITION, \tedge\}$-structure $\mathbb{B}\sqcup \mathbb{T}'$,
    where $\mathbb{T}'$ represents the $\{\tedge\}$-structure modelling the laminar tree $T$ of $(U,\mathcal{B})$.
    To complete the proof of the theorem, we provide an \MSO-transduction which outputs $\mathbb{T}$ given input $\mathbb{B}\sqcup \mathbb{T}'$.
    For the transduction we define an \MSO-formula $\phi_{\DEGENERATE}$ which selects the inner nodes of $T$ which are labelled $\DEGENERATE$,
    and a formula $\phi_{\cross}$ which selects 5-tuples of nodes $(t,x_1,x_2,x_3,x_4)$ such that either $x_1<_tx_2<_t x_3<_t x_4$ or $x_4<_t x_3<_t x_2<_t x_1$ or some cyclic shift of this.
    Recall, that the ordering $<_t$ for each inner node $t$ of the neighbours of $t$ are only unique up to inverting the order and cyclic shift.
    We use a functional symbol $\leafset(t,s)$ which is interpreted as the set $L(T_s^t)$ where $T_s^t$ is the connected component of $T-t$ containing $s$. 
    The functional predicate $\leafset$ clearly can be defined in \MSO as an element~$z$ belongs to $\leafset(t,s)$
    \ifof~$z$ is a leaf and the path from~$z$ to~$s$ does not contain~$t$. 

    Observe that by \autoref{thm:weakly-bipartitive tree} an inner node $t$ is labelled $\DEGENERATE$
    \ifof for every pair of neighbours $s_1,s_2$ of $t$ the bipartition $\big\{L(T^t_{s_1})\cup L(T^t_{s_2}),U\setminus \big(L(T^t_{s_1})\cup L(T^t_{s_2})\big)\big\}$ is in $\mathcal{B}$.
    We can define this in \MSO using the following formula:
    \begin{align*}
        \phi_{\DEGENERATE}&(x):=\forall y \forall z\Big(\big(\tedge(x,y)\land \tedge(x,z)\big)\implies\\
        &\BIPARTITION\big(\leafset(x,y)\cup \leafset(x,z), U\setminus (\leafset(x,y)\cup \leafset(x,z))\big)\Big).
    \end{align*}

    By \autoref{thm:weakly-bipartitive tree}, a set~$\mathcal{C}$ of neighbours of~$t$ is consecutive in~$<_t$ (understood as a cyclic order) 
    \ifof, the bipartition $\big\{\bigcup_{c\in \mathcal{C}}L(T^t_c),U\setminus\big(\bigcup_{c\in \mathcal{C}}L(T^t_c)\big)\big\}$ is a member of~$\mathcal{B}$. Furthermore, we can express that $x_1<_tx_2<_t x_3<_t x_4$ in $<_t$ (considered as a cyclic order) if there is a $<_t$-interval $\mathcal{C}$ which does not contain $x_4$, but contains two $<_t$-intervals $\mathcal{C}_1,\mathcal{C}_2$ such that $x_1,x_2\in \mathcal{C}_1$, $x_3\notin \mathcal{C}_1$ and $x_4\notin \mathcal{C}_2$ and $x_2,x_3\in \mathcal{C}_2$.
    Hence, we can define in \MSO the predicate $\cross$ as follows:
    \begin{align*}
        &\phi_\cross(t,x_1,x_2,x_3,x_4):= \Big[\lnot \DEGENERATE(t)\land \bigwedge_{i<j}x_i\not=x_j\land \bigwedge_{i}\tedge(t,x_i)\land
            \exists X,Y,Z
            \\*
            &\Big(\BIPARTITION(X)\land \BIPARTITION(Y)\land \BIPARTITION(Z)\land Y\subseteq X\land Z\subseteq X\land
            \leafset(t,x_1)\subseteq Y\setminus Z \land 
            \\*
            &\leafset(t,x_2)\subseteq Y\cap Z\land \leafset(t,x_3)\subseteq Z\setminus Y\land \leafset(t,x_4)\subseteq U\setminus X\Big)\Big].
    \end{align*}
    Note that any node which is not labelled $\DEGENERATE$ and for which there is a set $\mathcal{C}$ of at least two neighbours of $t$ for which $\big\{\bigcup_{c\in \mathcal{C}}L(T^t_c),U\setminus\big(\bigcup_{c\in \mathcal{C}}L(T^t_c)\big)\big\}$ is a member of~$\mathcal{B}$ must be labelled $\LINEAR$ implying that all nodes $t$ whose children are related through the $\cross$-relation must be labelled linear. Additionally, note that any node with less than three children which is not labelled $\PRIME$ must be labeled $\DEGENERATE$ implying that all nodes that are labelled linear have four (or more) children satisfying the $\cross$-relation. 
    
    Finally, the transduction forgets the predicate $\BIPARTITION$ to produce the output $\mathbb{T}$.
\end{proof}
Indeed, considering a bipartitive family $\mathcal{B}$ makes the use of nodes labelled $\LINEAR$ in the associated tree obsolete, as the following theorem states.
\begin{thmC}[{\cite[Theorem 4]{montgolfierThesis}}]
    Let $\mathcal{B}$ be a weakly-biaprtitive family and let \linebreak $(T,\lambda,(<_t)_{t\in \lambda^{-1}(\LINEAR)})$ be the weakly-bipartitive tree it induces. If $\mathcal{B}$ is bipartitive, then $\lambda^{-1}(\LINEAR)=\emptyset$ and all $<_t$ are empty.
\end{thmC}
As a consequence, we can capture any bipartitive family $\mathcal{B}$ by the simpler structure $(T,\lambda)$, called the \emph{bipartitive tree induced by $\mathcal{B}$} (or \emph{the bipartitive tree of $\mathcal{B}$}), where $\lambda$ is now a labelling function labelling each inner nodes of $T$ by either $\DEGENERATE$ or $\PRIME$.
Hence, we obtain the following as a corollary from \autoref{thm:transduce weakly-bipartitive tree}, where naturally the bipartitive tree $(T,\lambda)$ of a bipartitive family $\mathcal{B}$ is modelled by a $\{\tedge,\DEGENERATE\}$-structure $\mathbb{B}$ (the relation $\cross$ is not needed anymore because there are no nodes labelled $\LINEAR$ and hence we do not have to define linear orders $<_t$ on the children of nodes $t$ labelled $\LINEAR$).
\begin{cor}\label{cor:bipartite-tree}
    There exists a non-deterministic \CMSO[2]-transduction $\tau$ such that,
    for each bipartitive set system $(U,\mathcal{B})$ represented as the $\{\BIPARTITION\}$-structure $\mathbb{B}$,
    $\tau(\mathbb{B})$ is non-empty and every output in~$\tau(\mathbb{B})$
    is equal to some $\{\tedge,\DEGENERATE\}$-structure $\mathbb{T}$ representing the bipartitive tree $(T,\lambda)$ of~$(U,\mathcal{B})$.
\end{cor}

\subsection{Application to split decomposition}
In this section we consider the split decomposition introduced by Cunningham and Edmonds \cite{cunningham1982decomposition, cunninghamE1980}. We can
  understand a split of a graph as the opposite operation of a $1$-join, which informally takes two graphs $G$ and $H$ whose vertex-sets intersect in a single vertex $v$, and
  do a disjoint union of both while removing $v$ and adding the edges between every neighbour of $v$ in $G$ and every neighbour of $v$ in $H$. Applying the split operation recursively to a graph yields a decomposition called the
  split decomposition. Split decompositions are a well studied notion in structural and algorithmic graph theory (see for instance \cite{Courcelle20,GioanP12}). While Cunningham and
  Edmond's original work introduces split decompositions for undirected graphs \cite{cunninghamE1980}, split decompositions are considered for
  strongly connected directed graphs in \cite{cunningham1982decomposition}. Hence, throughout this section we let $G$ be a strongly connected directed
  graph.
A \emph{split} in $G$ is a bipartition $\{X,Y\}$ of $V(G)$ such that there are subsets $X^{\inNeighbours},X^{\outNeighbours}\subseteq X$, $Y^{\inNeighbours},Y^{\outNeighbours}\subseteq Y$ such that for all pairs of vertices $x\in X$ and $y\in Y$ the edge $xy\in E(G)$ \ifof either $x\in X^{\outNeighbours}$ and $y\in Y^{\inNeighbours}$ or $x\in Y^{\outNeighbours}$ and $y\in X^{\inNeighbours}$. Note that $X^{\inNeighbours}\cup X^{\outNeighbours}$ ($Y^{\inNeighbours}\cup Y^{\outNeighbours}$, \resp) need not necessarily be equal to $X$ ($Y$, \resp).
Additionally, if $G$ is strongly connected then $X^{\inNeighbours}\cup X^{\outNeighbours}\not=\emptyset$ and $Y^{\inNeighbours}\cup Y^{\outNeighbours}\not=\emptyset$. We call splits $\{X,Y\}$ for which either $|X|=1$ or $|Y|=1$ \emph{trivial splits}. A graph is considered \emph{prime with respect to splits} if it has only
trivial splits.
In a family $\mathcal{B}$ of splits we use the notion of \emph{strong splits} to coincide with the notion of strong bipartitions. See \autoref{fig:splitDec} (left) for an example of a strong split in a directed graph.

One can define split decompositions in greater generality, however, we only require the canonical split decomposition obtained by considering only strong splits in the following. Hence, consider a directed graph $G$ and let $\splitSet$ be the set of all splits of $G$. The following theorem shows that $\splitSet$ is weakly-bipartitive.

\begin{thmC}[\cite{cunninghamE1980,cunningham1982decomposition}]\label{thm:splitsBipartitive}
    For every strongly connected directed graph $G$  the family of splits $\splitSet$ is weakly-bipartitive.
\end{thmC}

 The (canonical) split decomposition of $G$ is the weakly-bipartitive tree of $\splitSet$. In the following we add additional structure to the split
   decomposition of $G$, obtaining the \emph{enriched split decomposition} of $G$, to be able to recover the graph $G$ from its split decomposition. Recall that
   $T_s^t$ denotes the connected component of $T-t$ which contains $s$ for any tree $T$. Let $T$ we the weakly-bipartitive tree of $G$. By construction of $T$, for every edge $uv$ of $T$ the bipartition $\{L(T_v^u),L(T_u^v)\}$ is a strong split of $G$. 
 For every edge $uv$ of $T$, we introduce two new vertices $\vec{uv}$ and $\vec{vu}$, called \emph{marker vertices}. Note that in particular every leaf of $T$ (and hence every vertex of $G$) corresponds to some marker vertex. For every inner node $u$ of $T$ we define a graph $G_u$, called a \emph{component} of the split decomposition, as follows. The set of vertices of $G_u$ is the set $\{\vec{uv}\mid v \text{ is a neighbour of }u\}$. Furthermore, there is an edge from $\vec{uv}$ to $\vec{uw}$ in $G_u$ if there are vertices $x\in L(T_v^u)$, $y\in L(T_w^u)$ such that $xy\in E(G)$. The enriched split decomposition of $G$ is the tuple $(H,F)$ where 
 \begin{itemize}
     \item $H$ is a directed graph consisting of the disjoint union of the graphs $G_u$ for every inner node $u$ of $T$ and
     \item $F$ is a symmetric edge relation containing  edges $\vec{uv}\vec{vu}$ for every edge $uv$ of $T$.
 \end{itemize} 
 For a clear distinction, we call the edges in $H$ \emph{$\mathsf{c}$-edges} and the edges in $F$ \emph{$\mathsf{t}$-edges}. Note that $\mathsf{t}$-edges correspond to the edges of $T$. See \autoref{fig:splitDec} for an illustration.

\begin{figure}
    \includegraphics[width=\textwidth]{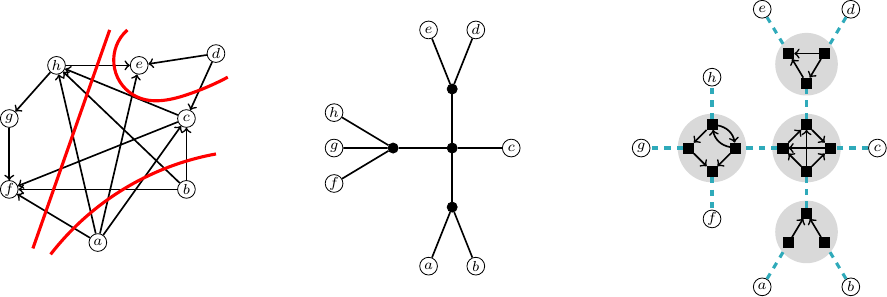}%
    \caption{A directed graph $G$ with $3$ strong splits (left). The weakly-partitive tree of $G$ (middle) and the split decomposition of $G$ where $\mathsf{c}$-edges are the thin, black edges, $\mathsf{t}$-edges are the fat, dashed, blue edges and the (non-singleton) components of the decomposition are signified by grey circles. Marker vertices are represented by black squares.%
    }%
    \label{fig:splitDec}
\end{figure}
\bigbreak

The following (standard) technical lemma is needed 
to show that a graph can be recovered from an enriched split decomposition (by an \MSO transduction).
\begin{lem}[See \eg~Lemma 2.10 in \cite{AdlerKK17}]\label{lem:edgesInSplitDecomposition}
    Let $G$ be a strongly connected, directed graph and $(H,F)$ the enriched split decomposition of $G$. Then 
       $uv\in E(G)$ \ifof there exists a directed path from $u$ to $v$ in $(H,F)$ on which $\textsf{c}$-edges and $\textsf{t}$-edges alternate; 
\end{lem}

To model enriched split decompositions as relational structures we use the relational vocabulary $\{\tedge,\cedge\}$ where $\tedge$ and $\cedge$ are binary relational symbols. An enriched split decomposition $(H,F)$ of a graph $G$ is the $\{\tedge,\cedge\}$-structure $\mathbb{H}$ with universe $U_\mathbb{H}=V(H)$, $\tedge_\mathbb{H}$ the set of $\textsf{t}$-edges in $(H,F)$ and $\cedge_{\mathbb{H}}$ the set of $\textsf{c}$-edges of $(H,F)$.

Since the nodes of $T$ correspond to the graphs $G_u$ which are the connected components of the enriched split decomposition $(H,F)$ after removing $F$ and the vertices of $G$ are the leaves of $T$, it is clear that \autoref{lem:edgesInSplitDecomposition} allows to $\MSO$-transduce the graph $G$ from its enriched split decomposition.

We use \autoref{lem:transduce laminar tree induced by weakly-bipartitive family} to transduce the canonical split decomposition from a graph.
\begin{thm}\label{thm:split}
    There exists a non-deterministic \CMSO[2]-transduction $\tau$
    such that for any strongly connected, directed graph $G$ represented as the $\{\edge\}$-structure $\mathbb{G}$;
    $\tau(\mathbb{G})$ is non-empty
    and every output in~$\tau(\mathbb{G})$ is equal to some $\{\tedge,\cedge\}$-structure $\mathbb{H}$
    representing the canonical split decomposition of $G$ (in particular, all outputs are isomorphic).
\end{thm}

\begin{proof}
    Let $G$ be a strongly connected, directed graph represented by the $\{\edge\}$-structure $\mathbb{G}$. We associate the following objects with $G$:
    \begin{itemize}
        \item let $\splitSet$ be the family of splits of $G$ and $\mathbb{B}$ the $\{\BIPARTITION\}$-structure modelling the weakly-bipartitive bipartition system $(V(G),\splitSet)$ with $U_\mathbb{B}=V(G)$;
        \item let $T$ be the laminar tree induced by the weakly-bipartitive family $\splitSet$ and let $\mathbb{T}$ be the $\{\tedge\}$-structure modelling it with $U_\mathbb{T}=V(T)$ and $L(T)=V(G)\subset U_\mathbb{T}$;
        \item let $(H,F)$ be the enriched canonical split decomposition of the graph $G$ represented by the $\{\tedge,\cedge\}$-structure $\mathbb{H}$.
    \end{itemize}
    Our \CMSO[2]-transduction is obtained by composing the following transductions:
    \begin{itemize}
        \item $\tau_1$: an \MSO-interpretation which outputs the $\{\edge,\BIPARTITION\}$-structure $\mathbb{G}\sqcup \mathbb{B}$ from $\mathbb{G}$;
        \item $\tau_2$: the non-deterministic \CMSO[2]-transduction from \autoref{lem:transduce laminar tree induced by weakly-bipartitive family} which outputs the $\{\edge,\BIPARTITION, \tedge\}$-structure $\mathbb{G}\sqcup \mathbb{B}\sqcup \mathbb{T}$ on input $\mathbb{G}\sqcup \mathbb{B}$;
        \item $\tau_3$: an \MSO-transduction which produces the $\{\tedge,\cedge\}$-structure $\mathbb{H}$ from $\mathbb{G}\sqcup \mathbb{B}\sqcup \mathbb{T}$.
    \end{itemize}
    To define $\tau_1$ it is sufficient to observe that there is an \MSO-formula $\phi_{\BIPARTITION}(X)$ with one monadic free-variable $X$ such that $\phi_{\BIPARTITION}$ is satisfied exactly when $\{X,U\setminus X\}$ is a split in $G$. As $\tau_2$ is constructed in \autoref{lem:transduce laminar tree induced by weakly-bipartitive family}, we are left with describing how to obtain $\tau_3$.
    We first need to find unique representatives for the vertices in $V(H)$. Recall that $V(H)$ consists of two marker vertices 
    $\vec{uv},\vec{vu}$ for every edge $uv$ in $T$.
    We use the following atomic transductions to ensure unique representatives for all vertices in $V(H)$. We first guess a colouring $R$ which identifies one non-leaf element of $\mathbb{G}\sqcup \mathbb{B}\sqcup \mathbb{T}$ as root. We use this root element to interpret an auxiliary binary relation $\parent$ in the usual way. Now we copy the structure $\mathbb{G}\sqcup \mathbb{B}\sqcup \mathbb{T}$ once introducing the binary relation $\copyV[1]$ in which $\copyV[1](x,y)$ indicates that $x$ is the copy of $y$. Assume that $\{s,t\}$ is an edge in $T$ corresponding to a strong bipartition in $\splitSet$ and $\parent(s,t)$ is satisfied. We use the original node $t$ to represent the marker 
    $\vec{st}$ and we use the copy of $t$ to represent the marker 
    $\vec{ts}$. As every node has a unique parent in a rooted tree, for each marker the chosen representatives are unique. To express this we use the predicate $\marker(x,y,z)$ which expresses that $x$ is the representative of the marker 
    $\vec{yz}$ defined as follows:
    $$\marker(x,y,z):= \Big(\parent(y,z)\land x=z\Big)\lor \Big(\parent(z,y)\land \copyV[1](x,y)\Big).$$
    We now filter the universe only keeping all representatives of markers 
    (recall that every original vertex $u$ of $G$ has a marker $\vec{uv}$, where $v$ is the neighbour of $u$ in $T$, representing it). We are left with defining the two edge relations $\tedge$ and $\cedge$. To do so, we use the previously defined functional predicate $\leafset(t,s)$ which is interpreted by the set $L(T_s^t)$.

    First note that markers  
    $\vec{st}$ and $\vec{s't'}$ are in the same component of $H$ \ifof 
    $s=s'$. By definition $\vec{st}\vec{st'}$ is a $\textsf{c}$-edge in $(H,F)$ \ifof there is an edge $uu'\in E(G)$ with 
    $u\in L(T^s_t)$ and $u'\in L(T^{s}_{t'})$.
    Consequently, we can express $\cedge$ as follows:
    \begin{align*}
        \phi_{\cedge}(x,x'):=&\exists t \exists t' \exists s \exists y \exists y' \Big( \tedge(s,t)\land \tedge(s,t')\land \marker(x,s,t)\land
        \\*
        &\marker(x',s,t')\land
        y\in \leafset(s,t)\land y'\in \leafset(s,t') \land E(y,y') \Big).  
    \end{align*}

    On the other hand, the
    $\tedge$-relation in $\mathbb{H}$ is the relation between pairs of markers and hence can be easily (re)-defined using the $\marker$ predicate by:
    $$\phi_{\tedge}(x,x'):=\exists y \exists z \big(\marker(x,y,z)\land \marker(x',z,y)\big).$$
    Finally, the transduction only keeps the relation $\tedge$ and $\cedge$. We remark that any choice of $R$ yields a valid output and hence we do not need to apply filtering.
  \end{proof}

    \begin{rem} Another notion of splits in directed graphs is defined in \cite{KanteR09} and yields in all directed graphs a weakly-bipartitive set family
      and a notion of enriched split decomposition. \autoref{thm:split} can be extended to this case also.
      \end{rem}

\subsection{Application to bi-join decomposition}\label{subsec:app-bijoin}
Bi-joins were introduced in \cite{MontgolfierR05,raoThesis} as a generalization of modules and splits in undirected graphs, and were used for instance in
\cite{LimouzyMR07} to decide isomorphism in some graph classes of rank-width at most $2$.  They were also used in \cite{raoThesis} to characterise some graph
classes, for instance ($C_5$,gem,co-gem,bull)-free graphs or (gem,co-gem)-free graphs. An extension to tournaments can be given and then those without prime
nodes correspond to locally transitive ones \cite{Bui-XuanHLM07}. We follow definitions and notations from
\cite{raoThesis}. To the best of our knowledge, bi-join decompositions have not been studied for directed graphs and therefore in this section we consider only undirected graphs. A \emph{bi-join} in an undirected graph $G$ is a bipartition $\{X,Y\}$ of $V(G)$ such that $\card{X}$, $\card{Y}\geq 1$ and there are
(possibly empty) subsets $X'\subseteq X$ and $Y'\subseteq Y$ such that
$\big(X'\times Y'\big) \cup \big( (X\setminus X')\times (Y\setminus Y')\big)\subseteq E(G)$, and there are no additional edges between $X$ and $Y$ in $G$. If
$X$ is a module in a graph $G$, then $\{X,V(G)\setminus X\}$ is a bi-join with $X'=\emptyset$, and if $V(G)\setminus X$ is complete to $X$, then $Y'$ is also
equal to the empty set.
As for splits, a bi-join $\{X,Y\}$ is considered \emph{trivial} if $\card{X}=1$ or $\card{Y}=1$, and a graph is considered \emph{prime with respect to bi-join} if it has only
trivial bi-joins. See \autoref{fig:simpleExampleBiJoin} and \autoref{fig:exampleBiJoin} for examples of bi-joins.
\begin{figure}
  \includegraphics[width=\textwidth]{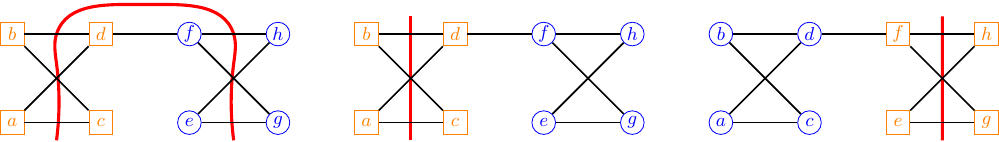}
  \caption{A simple graph $G$ with exactly three non-trivial bi-joins. Each bi-join is depicted using a thick red line. Vertices of set $X'$ and $Y'$ are signified by orange squares while vertices in $X\setminus X'$ and $Y\setminus Y'$ are blue circles.}
  \label{fig:simpleExampleBiJoin}
\end{figure}

\begin{figure}
  \includegraphics[width=\textwidth]{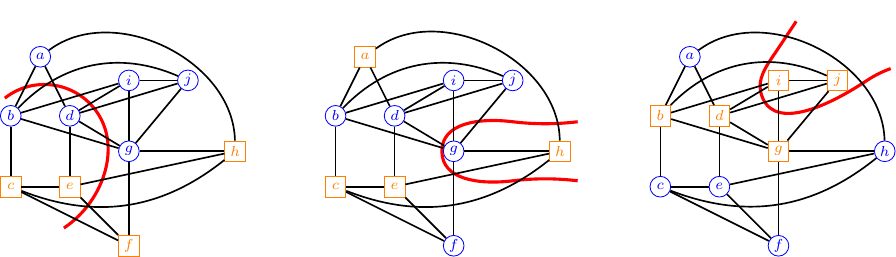}
  \caption{A graph $G$ with three non-trivial bi-joins. Each bi-join is depicted using a thick red line. Vertices of set $X'$ and $Y'$ are signified by orange squares while vertices in $X\setminus X'$ and $Y\setminus Y'$ are blue circles. Note that all three bi-joins depicted here are strong.}
  \label{fig:exampleBiJoin}
\end{figure}

Let us denote by $\joinSet$ the set of bi-joins of a graph $G$. 
A proof of the bipartitiveness of bi-joins can be found in
\cite{raoThesis}.

\begin{thmC}[\cite{raoThesis}]\label{thm:join-bipartitive} For every undirected graph $G$, the system $\joinSet$ is bipartitive. 
\end{thmC}

A corollary of \autoref{thm:join-bipartitive} and of \autoref{cor:bipartite-tree} is the existence of a \CMSO[2]-transduction taking as input an
undirected graph $G$ represented as the $\{\edge\}$-structure $\mathbb{G}$ and outputting the bipartitive tree representing the set of bi-joins of
$G$. We now propose a  \CMSO[2]-transduction that takes as input the structure $\mathbb{G}\sqcup \mathbb{T}$ where $\mathbb{T}$ is the structure representing
the bipartitive tree of $G$ and outputs a structure that is similar to the canonical split decomposition, from which one can
reconstruct the original graph in \MSO. It is worth mentioning that bi-joins refine in some sense modular and
  split-decomposition, and for computing decompositions based on complexity measures such as clique-width or rank-width one can first consider
  split-decomposition and then the bi-join decomposition of prime graphs \emph{w.r.t.} split-decomposition as bi-joins and splits intersect in very few
  cases. We refer to \cite[Section 3.7]{MontgolfierR05}, \cite{Courcelle20} and \autoref{sec:restricstedRankWidth} for more information.

\paragraph*{Transducing the Skeleton graph from a graph \cite{montgolfierThesis}} Let $G$ be an undirected graph. 
If $X$ is a subset of $V(G)$, let $\equiv_X$ be the binary relation on $X\times X$ where $x\equiv_X y$ if $N(x)\cap (V(G)\setminus X) = N(y)\cap (V(G)\setminus X)$. One
easily checks that $\equiv_X$ is an equivalence relation. Now if $\{X,Y\}$ is a bi-join, then $\equiv_X$ and $\equiv_Y$ both have at most two equivalence
classes. We remind that if $T$ is a tree and $t$ is a node of $T$ and $s$ a neighbour of $t$, then $T_s^t$ is the connected component of $T-t$ containing $s$,
and $L(T)$ is the set of leaves of $T$.

Let $T$ be the bipartitive tree of the set $\joinSet$ of the bi-joins of $G$. By construction of $T$, for every node $u$ of $T$ and every neighbour $v$
  of $u$, $\{L(T_v^u),L(T_u^v)\}$ is a bi-join, and then $\equiv_{L(T_v^u)}$ and $\equiv_{L(T_u^v)}$ have both at most two equivalence classes (see \autoref{fig:skeleton-graph} for an example). We denote by $\vec{uv}_1$ and $\vec{uv}_2$ the
  two equivalence classes of $\equiv_{L(T_v^u)}$ in what follows (we omit $\vec{uv}_2$ if there is only one). Observe that if $v$ is a leaf, then $\equiv_{L(T_v^u)}$, with $u$ its unique
  neighbour, has a single equivalence class reduced to $\{v\}$ and that we define as $\vec{uv}_1$. For every internal node $u$ of $T$, let us denote by $G_u$ the graph whose vertex set
  is $\{\vec{uv}_1\mid v$ is a neighbour of $u\}\cup \{\vec{uv}_2\mid v$ is a neighbour of $u, \equiv_{L(T^u_v)}$ has two equivalence classes$\}$, and there is an edge $Z_iZ_j$ in $G_u$ if there is $x\in Z_i$ and $y\in Z_j$ such that
  $xy\in E(G)$. Notice that there are $x\in Z_i$ and $y$ in $Z_j$ such that $xy\in E(G)$ if and only if every vertex in $Z_i$ is adjacent to every vertex in
  $Z_j$, and thus the graph $G_u$, for every node $u$, is unique. The \emph{Skeleton graph} associated with $T$ is 
  a triple $(H,F_1,F_2)$ where 
\begin{itemize}
\item $H$ is a graph with vertex set $\bigcup\limits_{u\in V(T)\setminus L(T)} V(G_u)$ and 
  \item edge set $\bigcup\limits_{u\in
      V(T)\setminus L(T)} E(G_u)$, 
     \item $F_1$ is the set of edges $\{\vec{uv}_i,\vec{vu}_j\}$, for $i,j\in [2]$ for which $u,v$ are inner nodes and there are $x\in \vec{uv}_i$ and $y\in \vec{vu}_j$ with $xy\in E(G)$. 
    \item $F_2$ is the set of pairs $\{\vec{uv}_1,\vec{uv}_2\}$ for which $u,v$ are inner nodes, $uv$ of $E(T)$ and $\equiv_{L(T_v^u)}$ has two equivalence classes.
  \end{itemize}
  We call the edges of $H$ \emph{$\mathsf{c}$-edges}, edges in $F_1$ \emph{$\mathsf{t}$-edges} and edges in $F_2$ \emph{$\mathsf{r}$-edges} and we naturally model the skeleton graph as $\{\tedge,\cedge,\redge\}$-structures. 

  It is worth mentioning that the Skeleton graph is uniquely defined from $T$, and since $T$ is unique, up to isomorphism, one can conclude that the
    Skeleton graph is unique, up to isomorphism. The Skeleton graph defined in \cite{montgolfierThesis} does not include $\mathsf{r}$-edges, but we need such
    edges to be able to recover the graphs $G_u$ from the Skeleton graph because $G_u$ without the $\mathsf{r}$-edges is not necessarily connected. See
    \autoref{fig:skeleton-graph} for an example.
    
    It is
    routine to show the following.

\begin{lemC}\label{lem:skeleton-to-graph-and-T} Let $G$ be a connected graph and let $(H,F_1,F_2)$ be its Skeleton graph. Then the following hold:
  \begin{enumerate}
  \item two vertices $x$ and $y$ of $G$ are adjacent \ifof there is a path between $x$ and $y$ in $(H,F_1,F_2)$ on which $\mathsf{c}$-edges and $\mathsf{t}$-edges alternate;
  \item the vertices of $G$ correspond precisely to the vertices of $(H,F_1,F_2)$ that are not incident to any $\mathsf{t}$-edges (recall that vertex $v$ corresponds to $\vec{uv}_1$ for the unique neighbour $u$ of $v$ in $T$).
  \end{enumerate}
\end{lemC}

\begin{figure}
  \includegraphics[width=\textwidth]{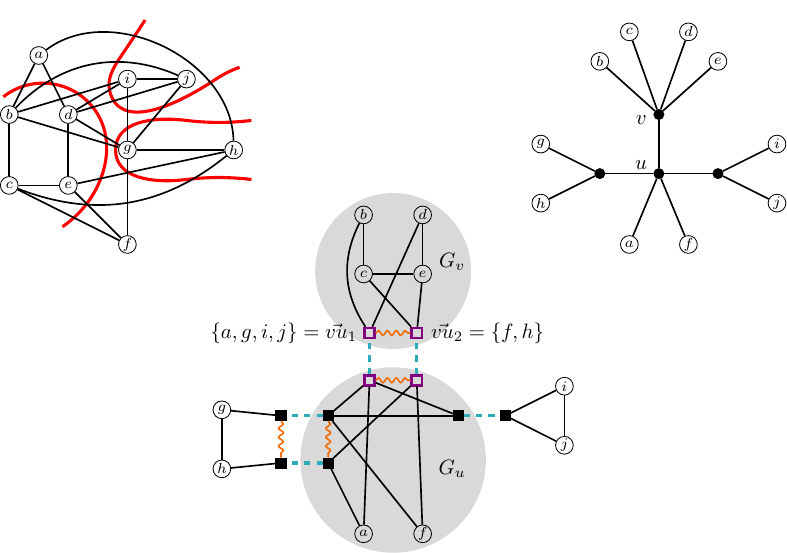}
  \caption{A graph $G$ with three non-trivial strong bi-joins (top left), the bipartitive tree of $\joinSet$ (top right)  and its Skeleton graph (bottom). Note that the three strong, non-trivial bi-joins of $G$ correspond to the three edges between inner nodes of the bipartitive tree - trivial bi-joins correspond to the edges incident to a leaf  (on the other hand, e.g. $\{\{a,f\}, \{b,c,d,e,g,h,i,j\}\}$ is not a strong bijoin of $G$ and there is no edge of the bipartitive tree that we can remove to obtain the bipartition $\{\{a,f\}, \{b,c,d,e,g,h,i,j\}\}$).
  Dashed edges are $\mathsf{t}$-edges, squiggly edges are $\mathsf{r}$-edges and the
    rest are $\mathsf{c}$-edges. Equivalence class vertices are represented by squares apart from the singleton equivalence classes corresponding to vertices of $G$. The four unfilled purple squares are the equivalence classes that correspond to the four copies of node $u$ in the proof of \autoref{thm:bi-join1}. Note that the edges of the graph $G$ correspond one-to-one to paths in the skeleton graph that alternate between $\mathsf{t}$-edges and  $\mathsf{c}$-edges.
    The graph depicted was taken from \cite[Figure 4.5]{raoThesis}. }
  \label{fig:skeleton-graph}
\end{figure}

It is routine to write an \MSO-formula checking that, between two vertices, there is a path that alternate $\mathsf{c}$-edges and
  $\mathsf{t}$-edges. 
  Hence, we can conclude that there is an \MSO-transduction which on input of the $\{\tedge,\cedge,\redge\}$-structures $\mathcal{H}$ modelling the skeleton graph of $G$ outputs the $\{\edge\}$-structure $\mathbb{G}$ modelling $G$.

We are now ready to construct the \CMSO[2]-transduction that computes the Skeleton graph of a graph.
\begin{thm}\label{thm:bi-join1}
    There exists a non-deterministic \CMSO[2]-transduction $\tau$
    such that for any connected graph $G$ represented as the $\{\edge\}$-structure $\mathbb{G}$,
    $\tau(\mathbb{G})$ is non-empty
    and every output in~$\tau(\mathbb{G})$
    is equal to some $\{\tedge,\cedge,\redge\}$-structure $\mathbb{H}$ representing the Skeleton graph of $G$.
\end{thm}

\begin{proof}
    Let $G$ be a connected graph represented by the $\{\edge\}$-structure $\mathbb{G}$, $(V(G), \joinSet)$ the bipartitive set system consisting of all bi-joins
    of $G$ represented by the $\{\BIPARTITION\}$-structure $\mathbb{B}$, $(H,F_1,F_2)$ its Skeleton graph represented by the $\{\tedge,\cedge,\redge\}$-structure $\mathbb{H}$ and $T$ represented by the $\{\tedge,\DEGENERATE\}$-structure $\mathbb{T}$, the bipartitive tree of $(V(G),\joinSet)$.

    Our first step is to transduce $\mathbb{B}$ from $\mathbb{G}$. For this it is sufficient to observe that we can easily define an \MSO-formula $\phi_{\BIPARTITION}(X)$ with one monadic free-variable $X$ which defines bi-joins $\{X,U\setminus X\}$ of $G$.

    Since $\mathbb{B}$ is bipartitive, we can apply \autoref{lem:transduce laminar tree induced by weakly-bipartitive family} to transduce $\mathbb{G}\sqcup
    \mathbb{B}\sqcup \mathbb{T}$ from $\mathbb{G}\sqcup \mathbb{B}$, where $\mathbb{T}$ is the structure representing the bipartitive tree of $G$. In the
    following we describe how to transduce $\mathbb{H}$ from $\mathbb{G}\sqcup \mathbb{B}\sqcup \mathbb{T}$.
      \begin{itemize}
        \item The first transduction chooses non-deterministically a node $r$ of $T$. We root $T$ at $r$ and defines the laminar set system $\mathbb{S}$ where
          $\SET(X)$ holds if there is a node
          $u$ with $X$ the set of leaves of the subtree of $(T,r)$ rooted at $u$. Let's denote by $L_u$ the set of leaves of the subtree rooted at $u$. 
        \item The second transduction makes four copies of each internal node $u$, adds $\mathsf{r}$-edges between copies $1$ and $2$, and between copies $3$ and $4$. The four copies of node $u$ will represent $\vec{uv}_1$, $\vec{uv}_2$, $\vec{vu_1}$ and $\vec{vu_2}$ for the parent $v$ of $u$.
        \item The next part of the transduction is necessary to assign equivalence classes to the four copies of a vertex. Note that the indexing of the
            equivalence classes above is chosen arbitrarily. In order to choose this assignment deterministically, we make use of the colouring transduction
            providing a representative for every inner node given in \autoref{cor:4 bi-colouring} and \autoref{lem:repr MSO}. Presume $x$ is the representative
            of inner node $u$ of $T$ and $y$ is a representative of the parent $v$ of $u$ chosen in such a way that $x$ and $y$ are in different parts of the
            bi-join represented by the edge $uv$. We can now choose that the first copy of $u$ represents the equivalence class containing $x$, the second copy
            represents the second equivalence class not containing $x$ and similarly for the part containing $y$. To implement this formally, we proceed by
            choosing the representatives. Recall, by \autoref{cor:4 bi-colouring} and \autoref{lem:repr MSO} there are four bi-colourings $(A_i,B_i)_{i\in [4]}$ of the elements of $\mathbb{G}$ identifying $V(T)$ and for each $i\in [4]$, there are
          \CMSO[2]-formulas $\repr_{A_i}(a,X)$ and $\repr_{B_i}(b,X)$ that are satisfied exactly when $X$ is, respectively, $A_i$-represented by $a$ and  $B_i$-represented
          by $b$. The third transduction just guesses non-deterministically the four bi-colouring $(A_i,B_i)_{i\in [4]}$ that identifies nodes of
          $(T,r)$ as a laminar tree of the laminar set system $\mathbb{B}$. Notice we can filter by accepting only the guesses with unique request and such that, for each node, there is a exactly one $i\in [4]$
          such that it is uniquely $A_i$-represented and uniquely $B_i$-represented.
           
          \item It remains now to create the $G_u$ graphs, for each internal node $u$. Assume $v$ is the parent of $u$, and let $a_v^u$ be the leaf such that either
            $\repr_{A_i}(a_v^u,L_v)$ or $\repr_{B_i}(a_v^u,L_v)$, and additionally $a_v^u\notin L_u$, and let $a_u^v$ be a leaf such that $\repr_{A_j}(a_u^v,L_u)$. Notice that $a_v^u$ exists because $v$ is the least common ancestor of the two leaves that $A_i$ and $B_i$-represent
            $v$, and if both do not belong to $L_u$, we deterministically define $a_v^u$ as the one that $A_i$-represents $v$. Copy $1$ of $u$ will play the role of the equivalence class of $\equiv_{L(T_v^u)}$ containing $a_v^u$. The copy $2$ will play the role of the
            equivalence class of $\equiv_{L(T_v^u)}$ not containing $a_v^u$. The copy $3$ will play the role of the equivalence class of $\equiv_{T_u^v}$ containing
            $a_u^v$, and the copy $4$ will play the role of the equivalence class of $\equiv_{L(T_{u}^v)}$ not containing $a_u^v$.
            The vertex set of $G_u$ will be copies $1$ and $2$ of $u$, and the copies $3$ and $4$ of each child of $u$ (see \autoref{fig:skeleton-graph} for an example). 
            For a child $w$ of $u$, we want to have an edge between the first copy of $u$ and the third copy of $w$ if there is an edge between $a_v^u$
              and $a_{w}^u$ in $\mathbb{G}$. Similarly, we want to have an edge between the second copy of $u$ and the fourth copy of $w$ if the second
              equivalence classes of $\equiv_{L(T_{w}^u)}$ and $\equiv_{L(T_v^u)}$ are non-empty and there are edges between those two equivalence classes.
            Hence, using the representative predicate from the previous transduction, we can write a \CMSO[2]-formula checking that the second equivalence class of $\equiv_{L(T_w^u)}$, $w$ a child of $u$, (resp. of $\equiv_{L(T_v^u)}$) is
            non-empty, and then whether there are edges between two equivalence classes. We can therefore write a \CMSO[2]-formula stating that two copies of
            vertices are related by a $\mathsf{c}$-edge and which copies should be included in the vertex set of $G_u$. 
          \item The fourth transduction adds the $\mathsf{t}$-edges. We add a $\mathsf{t}$-edge between a copy $i\in \{1,2\}$ of $u$ and a copy $j\in \{3,4\}$
            of $u$, if a vertex of the equivalence class the $i$th copy represents is adjacent to a vertex of the equivalence class the $j$th copy
            represents. Such edges can be \CMSO[2]-defined in the same way the $\mathsf{c}$-edges are defined, except they relate only particular copies of a
            same node.
          \end{itemize}
The composition of the four transductions above computes the Skeleton graph of $G$. Since the non-determinism comes only from the construction of
      $\mathbb{T}$ and the guessing of the bi-colouring, and we can filter after each such step, we guarantee that an output always satisfies the desired outcome.
\end{proof}

\section{Recognizability equals definability  in restricted classes of graphs of small rank-width}\label{sec:restricstedRankWidth}

We recall that a graph class is a set of graphs closed under isomorphism.  We refer to \cite{Boj23,CourcelleE12} for a review of the history and importance of the duality between the notions of recognizability and \CMSO-definability
  for graph classes. For word languages, a language $L$ is said to be recognizable if there is a monoid homomorphism from $L$ to a finite monoid, and this definition
  was extended to tree languages where monoid can be replaced with finite algebra, see for instance \cite[Section 1.5]{tata}. A tree language is roughly a set of finite well-formed terms, \ie{} terms
  constructed from a set of function symbols. An extension of the notion of recognizability to graph classes was proposed where graph operations\footnote{A
    graph operation preserves isomorphism and takes as input isomorphism classes.} are used to construct graphs from finite well-formed terms, which
  naturally yields a notion of recognizability similar to the one for tree languages, see for instance \cite{CourcelleE12}. A graph operation is a function
  which takes as inputs graphs and outputs graphs\footnote{In \cite{CourcelleE12} classes of graphs are classified into (infinitely many)
    sorts, and each operation takes as inputs (isomorphic classes of) graphs from sorts and outputs a (isomorphic class of) graph from a sort.}. Examples of graph
  operations are the HR graph operations defining tree-width or the VR graph operations defining clique-width \cite{CourcelleE12}.  

Following \cite{CourcelleE12,tata}, given a set $F$ of graph operations, a
  graph class $\mathcal{C}$ is said \defd{F-recognizable} if it is a union of equivalence classes of a congruence on $T(F)$ of finite index, where $T(F)$ is the
  set of finite well-formed terms over $F$. It is known that if $\mathcal{C}$ is a graph class of clique-width at most $k$, then any
  \CMSO-definable class $\mathcal{L} \subseteq \mathcal{C}$ is ${\rm CW}_k$-recognizable, where ${\rm CW}_k$ is the set of operations for clique-width at most $k$
  \cite{CourcelleE12}. The converse is still conjectured. However, if there is a \CMSO-transduction that takes as input a graph of clique-width at most $k$ and
  outputs a term in $T({\rm CW}_{k'})$ defining it, for some $k'$ depending only on $k$, then any ${\rm CW}_k$-recognizable subset of graphs of clique-width at most $k$ is
  \CMSO-definable. This is for instance a motivation for finding \CMSO-transductions building tree-like decompositions and was the route proposed by Courcelle
  and followed for instance in \cite{BojanczykP16,BojanczykGP21,CampbellGKKO25}.

We show in this section that our \CMSO-transduction for building canonical graph decompositions such as 
  split-decompositions allows to show that for some graph classes of small clique-width, ${\rm CW}_k$-recognizability and \CMSO-definability coincide, where $k$
  is the clique-width of the graph class. For simplicity, we will only sketch the steps of the proof for split-decomposition, the modular and bi-join
  decomposition cases are similar. We will furthermore use the complexity measure
  \defd{rank-width} equivalent to clique-width, which fits better with split-decomposition than clique-width. Indeed, if a graph class has clique-width at most $k$, then it has \defd{rank-width} at most $k$ and if it has rank-width at most
  $k$, then it has clique-width at most $2^{k+1}-1$ \cite{OumS06}. Moreover, there are \MSO-transductions relating rank-decompositions of width at most $k$
  with well-formed terms in $T({\rm CW}_{2^{k+1}-1})$, and well-formed terms in $T({\rm CW}_k)$ with rank-decompositions of width at most $k$
  \cite{CourcelleE12}.

Let's now define \defd{rank-width}. If $G$ is a graph, we denote by $A_G$,  the \emph{adjacency matrix of $G$}, a matrix over the binary field whose rows and columns are
  indexed by $V(G)$ and such that $A_G[x,y] =1$ only when $xy$ is an edge. For a graph $G$, we denote by $\rho_G:2^{V(G)}\to \mathbb{N}$, the \emph{cut-rank
    function of $G$} where $\rho_G(X)$, for $X\subseteq V(G)$, is the rank, over the binary field, of the submatrix of $A_G$ whose rows are indexed by $X$ and
  its columns by $V(G)\setminus X$. For a rooted tree $T$ and a node $u$ of $T$, let's denote by $N_T(u)$ the set of children of $u$.

A \emph{rank-decomposition} of a graph $G$ is a pair $(T,\delta)$ where $T$ is
  a rooted tree and $\delta:V(G)\to L(T)$ is a bijection from the vertex set of $G$ to the leaves of $T$. The \emph{cut-rank of $(T,\delta)$ over $G$} is defined as
  $$ \max\limits_{u\in V(T), X\subseteq N_T(u)} \left\{\rho_G\left(\bigcup_{v\in X}\delta^{-1}(L(T_v^u))\right)\right\}.$$
  The \emph{rank-width of a graph $G$} is the maximum cut-rank over all rank-decompositions of $G$. It is worth mentioning that our definition of
    rank-decomposition is, contrary to the usual definition in the literature, not binary and is rooted, but as one checks easily, any such rank-decomposition
    can be binarised in the usual way by taking any linear ordering of the children of a node with more than two children and be made unrooted, and still keeps
    the same width. These unranked rank-decompositions are necessary because built rank-decompositions from split-decompositions by \CMSO-transductions are
    necessarily unranked, otherwise \CMSO-formulas may order arbitrary sets. However, it is proved in \cite{CampbellGKKO25} that the computation of unranked
    rank-decompositions of small width is enough to prove that ${\rm CW}_k$-recognizability implies \CMSO-definability.

\begin{thm}\label{thm:cor-cmso-rank-dec} Let $\mathcal{C}$ be a class of graphs for which there is, for every $k$, a \CMSO-transduction $\varphi_{\mathcal{C},k}$, that takes as input a
  graph $G$ in $\mathcal{C}$ of rank-width at most $k$ and outputs a rank-decomposition of $G$ of width at most $f(k)$, for some computable function $f$. Then, for every positive integer
  $k$, there is a \CMSO-transduction
  $\psi_{\mathcal{C},k}$ that takes as input a graph $G$ of rank-width at most $k$, whose prime graphs with respect to split decomposition (resp. bi-join decomposition) belong to $\mathcal{C}$, and outputs a rank-decomposition of $G$ of width at most
  $g(k)$, for some computable function $g$.
\end{thm}

Graph classes where prime graphs \emph{w.r.t.} modular/split/bi-join decomposition are structured and a
  rank-decomposition can be computed by a \CMSO-transduction are studied in
\cite{BrandstadtDLM05,CiceroneS21,DucoffeP21a,HegerfeldK23}. For instance, the following.

\begin{cor}\label{cor:recog-definability} For all these graph classes recognizability and \CMSO-definability coincide : distance-hereditary graphs, (gem,co-gem)-free graphs, 
 ($P_5$,diamond)-free graphs, graphs of small modular-treewidth.
\end{cor}

\begin{proof} For every fixed $k$, there is an \MSO-transduction that takes a set of size at most $k$ and computes a star with $k$ leaves. This implies the case
  of distance-hereditary graphs as their split-decomposition do not contain prime graphs and (gem,co-gem)-free graphs as it
  is proved in \cite[Theorem 4.11, Lemmas 4.13, 4.14]{raoThesis} that their prime graphs, \emph{w.r.t.} bi-join decomposition, have bounded size. It is proved in
  \cite{BrandstadtDLM05} that prime ($P_5$,diamond)-free graphs \emph{w.r.t.} split-decomposition have either small size or belong to a finite set of graph
  classes, and for each there is a
  \CMSO-transduction computing a rank-decomposition. Small modular-treewidth \cite{HegerfeldK23} means that every prime graph \emph{w.r.t.} modular
  decomposition (hence \emph{w.r.t.} 
  split-decomposition) has bounded treewidth. The \CMSO-transduction from \cite{BojanczykP16} implies the case of graphs of small modular-treewidth.
\end{proof}

We observe that the graph classes mentioned in \autoref{cor:recog-definability} do not have bounded linear clique-width, \eg, distance-hereditary
  graphs \cite{KanteK18}, and then the corollary is not a consequence of \cite{BojanczykGP21,CampbellGKKO25}.  

We can now prove \autoref{thm:cor-cmso-rank-dec}. The transduction is the composition of three transductions: the first produces the enriched split-decomposition of the input thanks to
  \autoref{thm:split} or the skeleton graph thanks to \autoref{thm:bi-join1}, the
  second computes a rank-decomposition for each component of the $\mathsf{c}$-edges, and the last joins all these rank-decompositions using the
  $\mathsf{t}$-edges. However, the number of components induced by the $\mathsf{c}$-edges is not bounded, and so we need to compute all their
  rank-decompositions in parallel, which is possible
  thanks to the Parallel Application Lemma that we recall now.

Let~$\Sigma$ be a vocabulary and~$\mathbb{A}_1,\ldots,\mathbb{A}_n$ be disjoint $\Sigma$-structures.
Define the \emph{disjoint union} of structures~$\mathbb{A}_1,\ldots,\mathbb{A}_n$, denoted by~\emph{$\bigsqcup_{0<i\leq n}\mathbb{A}_i$},
as the following structure over the vocabulary~$\Sigma\cup\{\thicksim\}$
where~$\thicksim$ is a new binary relation name:
\begin{itemize}
  \item the universe of~$\bigsqcup_{0<i\leq n}\mathbb{A}_i$
    is the union of the universes of the $\mathbb{A}_i$'s
    (which are required to be disjoint);
  \item for each relation or predicate name from~$\Sigma$,
    its interpretation in~$\bigsqcup_{0<i\leq n}\mathbb{A}_i$
    is the union of its interpretations in each of the~$\mathbb{A}_i$'s;
  \item the interpretation of~$\thicksim$ in~$\bigsqcup_{0<i\leq n}\mathbb{A}_i$
    is the set of pairs of elements that originate from the same~$\mathbb{A}_i$.
\end{itemize}

\begin{lemC}[(Parallel Application Lemma \cite{BojanczykGP21})]
  \label{lem:parallell-transduction}
  Let~$\tau$ be a $\Sigma$-to-$\Gamma$ \CMSO-transduction. Then there is a 
  $(\Sigma\cup\{\thicksim\})$-to-$(\Gamma\cup\{\thicksim\})$ \CMSO-transduction~$\hat{\tau}$ such that, for every
  sequence~$\mathbb{I}_1,\ldots,\mathbb{I}_n$ of $\Sigma$-structures and every sequence~$\mathbb{O}_1,\ldots,\mathbb{O}_n$ of $\Gamma$-structures, we have
  $(\bigsqcup_{0<i\leq n}\mathbb{I}_i,\bigsqcup_{0<i\leq n}\mathbb{O}_i)\in\hat{\tau}$ \ifof there exists a permutation~$\pi$ of~$[n]$
  such that $(\mathbb{I}_i,\mathbb{O}_{\pi(i)})\in\tau$ for all~$i\in[n]$.
\end{lemC}

We are now ready to give the proof of \autoref{thm:cor-cmso-rank-dec}.

\begin{proof}[Proof of \autoref{thm:cor-cmso-rank-dec}] Let's first consider the transduction for split-decomposition. Let $G$ be a connected graph represented by the $\{\edge\}$-structure $\mathbb{G}$. 
    The transduction is the composition of the following.
    \begin{enumerate}
    \item The first transduction takes as input $\mathbb{G}$ and outputs, thanks to \autoref{thm:split}, a $\{\cedge,\tedge\}$-structure $\mathbb{H}$ which is the canonical split decomposition
      of $G$. 
    \item The second transduction takes as input $\mathbb{H}$ and outputs a $\{\cedge,\tedge,\thicksim\}$-structure $\mathbb{H}'$ which is $\mathbb{H}$
        equipped with the equivalence relation $\thicksim$ where $x\thicksim y$ if and only if they belong to the same connected component after removing
        $\mathsf{t}$-edges. Recall that the set of $\mathsf{t}$-edges is the same as the set of edges of the laminar tree reprensenting the set of strong splits
        of $G$. Moreover, each obtained component is either a star or a complete graph or a prime graph. 
      \item Because any subgraph of $\mathbb{H}$ induced by an equivalence class of $\thicksim$ that is not prime is either a star or a complete graph, one can trivially extend
        $\varphi_{\mathcal{C},k}$ to $\varphi'$ that outputs a rank-decomposition for every subgraph of $\mathbb{H}$ induced by an equivalence class of $\thicksim$,
        and of cut-rank $1$ if not prime, otherwise of cut-rank at most $f(k)$. 
      \item By the Parallel Application Lemma (\autoref{lem:parallell-transduction}), there is a \CMSO-transduction $\hat{\varphi'}$ that takes as input $\mathbb{H}'$ and outputs
        rank-decompositions of the subgraphs of $\mathbb{H}'$ induced by the equivalence classes of $\thicksim$.
      \item The last transduction takes all rank-decompositions output by $\hat{\varphi'}$ and fuses any two leaves that are related by a $\mathsf{t}$-edge in
        $\mathbb{H}$. The root of the rank-decomposition of one component can be chosen arbitrarily as a root.
        \end{enumerate}
        It is routine to check that the composition, in the same order, of the above transductions, outputs a rank-decomposition of $G$ as the $\mathsf{t}$-edges form a tree and
        because each of them yields a correct output when given a correct input. Because the rank-decompositions output by $\hat{\varphi'}$ have cut-rank at
        most $f(k)$ on prime subgraphs and of cut-rank $1$ otherwise, it is easy to check that the cut-rank of the built rank-decomposition of $G$ has
        cut-rank at most $f(k)$.

        If $G$ is not a connected graph, then we apply the above \CMSO-transduction on each connected component thanks again to the
        Parallel Application Lemma (\autoref{lem:parallell-transduction}), and then a second deterministic transduction adds a new node whose neighbour in the rank-decomposition of each connected
        component is its root.

        The proof for bi-join decompositions works very similar, except we call \autoref{thm:bi-join1} to compute the Skeleton graph instead of an enriched
        split-decomposition. Let $T$ be the bipartitive tree of $\joinSet$. The Skeleton graph of $G$ associates with each node $u$ of $T$ a graph $G_u$
          (see \autoref{subsec:app-bijoin}). For every node $u$ of $T$, a \defd{characteristic graph of $u$} is a subgraph of $G_u$ obtained by choosing, for
          each edge $uv$ of $T$, at most one vertex among $\{\vec{uv}_1,\vec{uv}_2\}$. It is proved in \cite[Lemma 4.7]{raoThesis} that characteristic graphs
          are prime for prime nodes, otherwise they are degenerate. For an edge $uv$ of $T$, a characteristic graph $H_u$ for $u$ and a characteristic graph
          $H_v$ for $v$ are \defd{compatible} if the vertex choosen for $uv$ in $H_u$ and the vertex choosen for $uv$ in $H_v$ correspond to two adjacent
          equivalence classes of the bipartition. It is proved in \cite[Section 6]{raoThesis} that, if characteristic graphs associated with the nodes of the
          bi-join decomposition have clique-width at most $k$ and characteristic graphs of adjacent nodes are all compatible, then one can construct a term in
          $T({\rm CW}_{8k})$ of $G$, construction similar to the one linking modular decomposition and clique-width \cite{CourcelleO00}.  It is not
          hard then to check that the rank-decomposition constructed as in the split decomposition case is a rank-decomposition of $G$ of cut-rank at most $2k$
          if all characteristic graphs have rank-width at most $k$. Therefore, the transduction for bi-join decomposition is similar to the one for
          split-decomposition except, for computing the rank-decomposition, it selects vertices to remove from $G_u$s and must check compatibility. This
          finishes the proof.
  \end{proof}

\section{Conclusion}
We provide transductions for obtaining tree-like graph decompositions such as modular decompositions, cotrees, split decompositions and bi-join decompositions
from a graph using \CMSO. This improves upon results of Courcelle \cite{Courcelle96,Courcelle99,Courcelle06} who gave such transductions for ordered graphs. It is worth
mentioning that \autoref{thm:transduce weakly-bipartitive tree} can be also used to \CMSO-transduce canonical decompositions of other structures such as Tutte's
decomposition of matroids or generally 
split-decompositions of submodular functions \cite{cunningham1982decomposition} or modular decompositions of $2$-structures \cite{ehrenfeuchtHR1999} or of hypergraphs
\cite{HabibMMZ22}. As shown by the application given in \cite{Courcelle06} for transducing Whitney's isomorphism class of a graph, a line of research is to further
investigate which structures can be \CMSO-transduced from a graph or a set system by using the transductions from \autoref{thm:summary}. Also, naturally, the
question arises whether counting is necessary or whether \MSO is sufficient to transduce such decompositions.

Furthermore, it is known that if a family of set systems $\mathfrak{S}$ is obtained from a family $\mathfrak{F}$ of $\Sigma$-labelled trees, for some
  finite alphabet $\Sigma$, by a \CMSO-transduction $\tau$, then for any \CMSO formula $\varphi$, the set
  $\{T\in \mathfrak\mid \tau(\mathbb{T})\models \varphi\}$ is a recognizable subfamily of $\mathfrak{F}$ \cite{CourcelleE12,FunkMN22}.  However, the converse, \ie,
  whether recognizable families of $\mathfrak{S}$ are \CMSO-definable is open, and
  and was originally conjectured in \cite{TMSOLOG1} for graph classes of small tree-width. Even though Courcelle's conjecture is now settled \cite{BojanczykP16},
  the case of graphs that are images of trees by \CMSO-transductions, equivalently graph classes of bounded clique-width \cite{CourcelleE12}, is still open. 
  In \autoref{sec:restricstedRankWidth} we have reduced this question to prime graphs \emph{w.r.t.} split-decomposition (or bi-join decomposition), and this allows to
  obtain new graph classes of small clique-width where recognizability and \CMSO-definability are equivalent. 

\bibliographystyle{alphaurl}
\bibliography{finalBib.bib}

\end{document}